\documentclass[letterpaper,12pt]{article}
\usepackage{amsfonts}
\usepackage{amsmath, amsthm, amssymb, bbm}
\usepackage{natbib}
\usepackage{amsfonts}
\usepackage{amsthm}
\usepackage{graphics}
\usepackage{bm}
\usepackage{float}
\usepackage{appendix}
\usepackage{comment}
\usepackage[multiple]{footmisc}
\usepackage{rotating}
\usepackage{subcaption}
\usepackage{booktabs}
\usepackage{color}
\usepackage{appendix}
\usepackage{algorithm}
\usepackage{algpseudocode,enumerate} 
\input{ee.sty}
\RequirePackage[colorlinks,citecolor=blue,urlcolor=blue,linkcolor=blue]{hyperref}

\newtheorem{remark}{Remark} 
\newtheorem{theorem}{Theorem}

\newtheorem{lemma}{Lemma}

\newcommand{\argmin}{\operatorname*{argmin}}

\begin{document}

\title{\bf Factor-Augmented Panel Regressions and Variance-Weighted Treatment
Effects\thanks{%
Financial support from the Dutch Research Council (NWO) under research grants $451-17-002$ and VI.Vidi.231E.030 is gratefully acknowledged by the first author. The second author thanks Patrick Gagliardini and Olivier Scaillet for organizing  the 2025 EC$^2$ Conference in Lugano and for the invitation to present this paper  as the ET Lecture, and Peter Phillips and the journal \emph{Econometric Theory} 
for sponsoring the ET Lecture series.
}}
\author{\setcounter{footnote}{2}Art{$\bar {\rm u}$}ras Juodis\thanks{%
University of Amsterdam and Tinbergen Institute, \texttt{a.juodis@uva.nl}.  }
\and 
 Martin Weidner%
\thanks{%
University of Oxford and Nuffield College, \texttt{martin.weidner@economics.ox.ac.uk} } }
\date{April 2026}

\maketitle
\thispagestyle{empty}
\setcounter{page}{0}

\begin{abstract}
\noindent  
We revisit panel regressions with unobserved heterogeneity through the lens of variance-weighted average treatment effects. Building on established results for cross-sectional OLS and one-way fixed effects panels, we show that two-way panel estimators with latent factors, specifically the principal components estimator of \citet{GreenawayMcGrevy201248} and the interactive fixed effects estimator of \citet{bai2009panel}, also converge to interpretable estimands under fully nonparametric assumptions. Both estimators consistently estimate the same variance-weighted average of unit-time-specific treatment effects, where the weights are proportional to the conditional variance of the regressor given the unobserved heterogeneity. The result requires the number of estimated factors to grow with the sample size and applies to the single regressor case. We discuss the challenges that arise when extending to multiple regressors and to inference.
\end{abstract}

\newpage

\section{Introduction}

Linear panel regressions with interactive fixed effects have become a standard tool 
in empirical economics since the seminal contributions of \citet{pesaran2006estimation} 
and \citet{bai2009panel}. These estimators are designed for panels of $n$ units 
observed over $T$ time periods, and allow for unobserved unit- and time-specific 
heterogeneity that enters the outcome model multiplicatively, generalizing the additive 
two-way fixed effects structure. A large literature has since developed around these 
methods, studying their asymptotic properties, the selection of the number of factors, 
and extensions to nonlinear models. However, the theoretical properties of these 
estimators are almost exclusively derived under parametric assumptions: the data 
generating process is assumed to follow a linear factor model of the form 
$Y_{it} = X_{it}\beta + \lambda_i'f_t + \varepsilon_{it}$, and consistency is established 
for the true parameter $\beta$ of that model. A natural and largely unanswered question 
is what these estimators converge to when the parametric model is misspecified or when 
no parametric structure is imposed at all.

This paper answers that question for the principal components (PC) estimator 
of \citet{GreenawayMcGrevy201248} and the interactive fixed effects (IFE) estimator of 
\citet{bai2009panel}. We show that both estimators have well-defined and interpretable probability limits under 
fully nonparametric assumptions on the data generating process. The key insight comes 
from combining the factor model literature with results on variance-weighted treatment 
effects from the cross-sectional literature. \citet{Angrist1998} showed that OLS with 
observed controls estimates a variance-weighted average of heterogeneous treatment 
effects, where the weights are proportional to the conditional variance of the regressor 
given the controls. We show that the same interpretation extends to the two-way panel 
setting with unobserved heterogeneity: both the PC and IFE estimators converge to a 
variance-weighted average of unit-time-specific treatment effects, where each weight 
is proportional to the conditional variance of the regressor given the unobserved 
unit- and time-specific heterogeneity. This result does not require any parametric restrictions 
on the conditional expectation model, but the number of estimated factors $R_{nT}$ is 
required to grow with the sample size. The reason this works is that a growing number 
of principal components effectively turns the truncated singular value decomposition 
(SVD) of the regressor matrix into a nonparametric estimator of its conditional mean 
given the unobserved heterogeneity: classical results in approximation theory show 
that any sufficiently smooth bivariate function admits a Hilbert-Schmidt expansion 
with rapidly decaying singular values, so that the rank-$R_{nT}$ truncation 
approximates the conditional mean arbitrarily well as $R_{nT} \to \infty$ 
\citep{birman1977estimates}. By contrast, other popular approaches such as the common correlated effects (CCE) 
estimator of \citet{pesaran2006estimation} or the standard two-way fixed effect (TWFE) estimator generally 
do not share this property.

Our result can be seen as the large $n,T$ analogue of findings in the DiD and TWFE 
literature that have attracted considerable attention since 2018. \citet{2018arXiv180308807D}, 
\citet{Goodman2021}, \citet{SUN2020}, and \citet{CALLAWAY2020}, among others, showed 
that TWFE estimators with staggered treatment adoption can assign negative weights to 
some treatment effects, and proposed alternative estimators that deliver non-negative 
weights. In that literature, parallel trends assumptions are required throughout because 
$T$ is treated as fixed, and the identification problem is intrinsically harder with 
limited time periods. Our setting is different: with $T$ growing, the factor structure 
can be estimated nonparametrically, and the variance-weighted average can be recovered 
without any parallel trends assumption. Nevertheless, the conceptual message is similar: 
pooled regression estimators target specific weighted averages of heterogeneous effects, 
and non-negativity of those weights is essential for interpreting the estimates.

There is also a growing literature that, while allowing for additive unobserved 
intercepts, relaxes the assumption of common slope coefficients, see e.g.\ 
\citet{LU2023694} and \citet{KeaneNeal2020}. Some papers allow for rich heterogeneity 
in both slopes and intercepts, see e.g.\ \citet{chernozhukov2019inference} and 
\citet{unknown2024estimation}. All of these approaches impose parametric structure on 
how the unit-time-specific treatment effects vary with the unobserved heterogeneity. 

In the extension section,
we informally discuss estimating user-specified weighted averages of these effects 
without such restrictions. We also discuss the limitation that our 
variance-weighted average interpretation is specific to the case of a single regressor. 
When multiple regressors are included, contamination bias arises as documented by 
\citet{10.1257/aer.20221116} for cross-sectional regressions and by \citet{de2023two} 
for TWFE regressions with several treatments. Finally, we explain that due to the nonparametric nature of our results, inference on the 
variance-weighted average is generally not possible without additional assumptions, as the asymptotic distribution of both estimators is dominated by 
approximation bias.

The remainder of the paper is structured as follows. Section~\ref{section:estimands} 
defines the variance-weighted estimands for cross-sectional and panel data settings. 
Section~\ref{section:estimators} introduces the PC and IFE estimators, discusses the 
difficulties of controlling for unobserved two-way heterogeneity, and states the main 
consistency results. Section~\ref{section:discussion} contains further discussion and 
extensions, covering inference, estimation of targeted treatment effects with 
user-specified weights, and the case of multiple regressors. 
Section~\ref{section:montecarlo} explores the finite sample properties of the methods 
using simulations. Section~\ref{section:conclusions} concludes. 
Appendix~\ref{section_appendix:proofs} contains all proofs.

\emph{Notation.} For any matrix $A$ we write $\|A\|$ for the spectral norm and $\|A\|_2$ 
for the Frobenius norm.
 
 \section{Variance-weighted average treatment effects}
\label{section:estimands}

In this section we introduce the variance-weighted treatment effect framework. We begin with the cross-sectional setting, where the key ideas are most transparent. We then extend the framework to panel data with unobserved heterogeneity. The results reviewed here are well established in the literature. Our purpose in this section is purely to set up the notation and conceptual apparatus needed for the two-way panel setting in Section~\ref{section:estimators}.

\subsection{Cross-sectional setting}
\label{ssection::cross_sectional}

Consider a random sample $(Y_i, X_i, C_i)$, $i = 1, \ldots, n$. Here $Y_i \in \mathbb{R}$ is an outcome variable, $X_i \in \mathbb{R}$ is a scalar regressor of interest, and $C_i \in \mathbb{R}^p$ is a vector of control variables. We are interested in the effect of $X_i$ on $Y_i$ while controlling for $C_i$. We begin with the simplest case where no controls are present. Assume only that $(Y_i, X_i)$ are identically distributed with finite second moments. Let $\widehat{\beta}$ denote the coefficient on $X_i$ from the OLS regression of $Y_i$ on $(1, X_i)'$. Then
\begin{align}
    \widehat{\beta} \to_P \beta^* := \frac{{\rm Cov}\left(Y_i,\, X_i\right)}{{\rm Var}\left(X_i\right)}
    \label{BasicOLS}
\end{align}
as $n \to \infty$.\footnote{When the regressor is binary, $X_i \in \{0, 1\}$, this simplifies to $\beta^* = \mathbb{E}\left(Y_i \,\big|\, X_i = 1\right) - \mathbb{E}\left(Y_i \,\big|\, X_i = 0\right)$, which has a causal interpretation as an average treatment effect if $X_i$ is randomly assigned.}
Now suppose we wish to control for the observed covariates $C_i$. For each value of the controls, we can define the conditional linear projection coefficient
\begin{align}
    B(c) := \frac{{\rm Cov}\left(Y_i,\, X_i \,\big|\, C_i = c\right)}{{\rm Var}\left(X_i \,\big|\, C_i = c\right)}.
    \label{eq::DefBC}
\end{align}
This measures the linear relationship between $Y_i$ and $X_i$ holding $C_i$ fixed.\footnote{Again, for binary $X_i$, $B(c)$ can be expressed as the conditional average treatment effect $B(c) = \mathbb{E}\left(Y_i \,\big|\, X_i = 1,\, C_i = c\right) - \mathbb{E}\left(Y_i \,\big|\, X_i = 0,\, C_i = c\right)$.} In general, $B(c)$ varies with $c$, that is, the effect of $X_i$ on $Y_i$ is heterogeneous.
Rather than estimating $B(c)$ for each value of $c$ separately, it is natural to consider a pooled estimand that summarizes the overall relationship. To this end, we define the orthogonal decomposition
\begin{align}
    X_i = X_i^\parallel + X_i^\perp, \qquad X_i^\parallel := \mathbb{E}\left(X_i \,\big|\, C_i\right), \qquad X_i^\perp := X_i - X_i^\parallel.
    \label{eq::DefXperpCS}
\end{align}
Here $X_i^\parallel$ is the conditional mean of $X_i$ given $C_i$, while $X_i^\perp$ is the residual variation in $X_i$ that remains after conditioning on the controls. Note that $\mathbb{E}\left(X_i^\perp \,\big|\, C_i\right) = 0$ by construction. The pooled estimand is now defined as
\begin{align}
    \beta^* := \frac{{\rm Cov}\left(Y_i,\, X_i^\perp\right)}{{\rm Var}\left(X_i^\perp\right)} = \frac{\mathbb{E}\left(Y_i X_i^\perp\right)}{\mathbb{E}\left[\left(X_i^\perp\right)^2\right]}.
    \label{eq::DefBetaStarCS}
\end{align}
This generalizes the definition in equation \eqref{BasicOLS} to account for controls. It is the population regression coefficient from regressing $Y_i$ on the residualized regressor $X_i^\perp$. A key observation is that $\beta^*$ can be expressed as a weighted average of the conditional effects $B(c)$. We have\footnote{By the law of iterated expectations we have $\mathbb{E}\left(Y_i X_i^\perp\right) = \mathbb{E}\left[\mathbb{E}\left(Y_i X_i^\perp \,\big|\, C_i\right)\right] = \mathbb{E}\left[{\rm Cov}\left(Y_i,\, X_i \,\big|\, C_i\right)\right]$ and $\mathbb{E}\left[\left(X_i^\perp\right)^2\right] = \mathbb{E}\left[\mathbb{E}\left(\left(X_i^\perp\right)^2 \,\big|\, C_i\right)\right] = \mathbb{E}\left[{\rm Var}\left(X_i \,\big|\, C_i\right)\right]$, where we used that $\mathbb{E}\left(X_i^\perp \,\big|\, C_i\right) = 0$.} 
\begin{align}
    \beta^* = \frac{\mathbb{E}\left[{\rm Var}\left(X_i \,\big|\, C_i\right) \cdot B(C_i)\right]}{\mathbb{E}\left[{\rm Var}\left(X_i \,\big|\, C_i\right)\right]} = \mathbb{E}\left[W^*(C_i) \cdot B(C_i)\right],
    \label{eq::BetaStarWeightedAvg}
\end{align}
where the weights are given by
\begin{align}
    W^*(c) := \frac{{\rm Var}\left(X_i \,\big|\, C_i = c\right)}{\mathbb{E}\left[{\rm Var}\left(X_i \,\big|\, C_i\right)\right]}.
    \label{eq::DefWeightsCS}
\end{align}
These weights are non-negative by construction and satisfy $\mathbb{E}\left[W^*(C_i)\right] = 1$. Thus $\beta^*$ is a convex combination of the conditional effects $B(c)$. Observations receive larger weight when the conditional variance of $X_i$ given $C_i$ is larger. This is the variance-weighted average treatment effect interpretation of regression estimands emphasized by \citet{Angrist1998} and subsequent literature. See also \citet{angrist2008mostly}, \citet{abadie2018econometric}, \citet{sloczynski2020interpreting}, and \citet{10.1257/aer.20221116}.

An important feature of $\beta^*$ is that it can be consistently estimated by OLS without ever explicitly having to estimate $B(c)$. Suppose that either of the following two conditions holds\footnote{Note that condition \eqref{eq::LinearCondMean} (a) and (b) are both automatically satisfied when $C_i$ consists of a saturated set of dummy variables. This observation is particularly relevant for the panel data setting in the next subsection, where the unobserved controls $U_i$ can be thought of as unit-specific dummies.}
\begin{equation}
\label{eq::LinearCondMean}
\begin{aligned}
    &\text{(a)} \quad \mathbb{E}\left(X_i \,\big|\, C_i\right) = \kappa_1 + \kappa_2' \, C_i, \quad \text{for some } \kappa_1 \in \mathbb{R}, \, \kappa_2 \in \mathbb{R}^p, \\
    &\text{(b)} \quad \mathbb{E}\left(Y_i - \beta^* \, X_i \,\big|\, C_i\right) = \rho_1 + \rho_2' \, C_i, \quad \text{for some } \rho_1 \in \mathbb{R}, \, \rho_2 \in \mathbb{R}^p.
\end{aligned}
\end{equation}
Then the OLS estimator $\widehat{\beta}$ on $X_i$ from regressing $Y_i$ on $(1, X_i, C_i')'$ consistently estimates $\beta^*$ (here we assume existence of second moments throughout). To see this, let $\widehat{X_i^\perp}$ denote the OLS residual from regressing $X_i$ on $(1, C_i')'$. Under condition \eqref{eq::LinearCondMean}(a), $\widehat{X_i^\perp}$ is a consistent estimator of $X_i^\perp$. The Frisch-Waugh-Lovell theorem then implies that
\begin{align}
    \widehat{\beta} = \frac{\sum_{i=1}^n Y_i \widehat{X_i^\perp}}{\sum_{i=1}^n \left(\widehat{X_i^\perp}\right)^2}  \to_P \beta^*
\end{align}
as $n \to \infty$.

Equally,
under condition \eqref{eq::LinearCondMean}(b)
and with $\beta^*$ defined in \eqref{eq::DefBetaStarCS},
we have $Y_i - \beta^* X_i - \rho_1 - \rho_2' C_i$ uncorrelated with $(1, X_i, C_i')'$, so $\beta^*$ is the population OLS coefficient on $X_i$.\footnote{A more standard way to express condition \eqref{eq::LinearCondMean}(b) is as follows. Assume that $Y_i = \rho_1 + \beta X_i + \rho_2' C_i + \varepsilon_i$ for some $\rho_1$, $\beta$, $\rho_2$ such that $\mathbb{E}(\varepsilon_i \,|\, C_i) = 0$ and $\mathbb{E}(\varepsilon_i X_i) = 0$. Then $\beta = \beta^*$ (assuming $\beta^*$ is well-defined), and the OLS coefficient $\widehat{\beta}$ on $X_i$ from regressing $Y_i$ on $(1, X_i, C_i')'$ consistently estimates $\beta^*$.}
Admittedly, condition \eqref{eq::LinearCondMean}(b) may seem redundant here, since both \eqref{eq::LinearCondMean}(a) and \eqref{eq::LinearCondMean}(b) justify the same OLS estimator. However, the dual justification will be important in Section~\ref{section:estimators}, where the generalization of each condition corresponds to a different estimation approach.

More generally, one may wish to estimate weighted averages of $B(c)$ with user-specified weights,
\begin{align}
    \beta_w = \mathbb{E}\left[w(C_i) \cdot B(C_i)\right],
\end{align}
for some weighting function $w(\cdot)$. When $C_i$ takes discrete values, this can be done by first estimating $B(c)$ separately for each value $C_{i}=c$ and then forming the weighted average given the weights $w(c)$. However, establishing consistency for such estimators is generally more difficult than for $\widehat{\beta}$ above. This is particularly the case when the overlap condition ${\rm Var}\left(X_i \,\big|\, C_i = c\right) > 0$ fails for some values of $c$. We return to this issue in Section~\ref{ssection:targeted}.

\subsection{Panel data with one-way unobserved heterogeneity}
\label{ssection::unobserved_oneway}

We now extend the framework to panel data where the controls are unobserved. Consider a sample $(Y_{it}, X_{it})$ for $i = 1, \ldots, n$ and $t = 1, \ldots, T$. Here $Y_{it} \in \mathbb{R}$ is an outcome and $X_{it} \in \mathbb{R}$ is a scalar regressor. We assume that for each unit $i$ there exists an unobserved confounder $U_i$ (possibly multi-dimensional) that may affect both $Y_{it}$ and $X_{it}$. Formally, we write\footnote{%
We do allow for a DGP of the form $Y_{it} = h(X_{it}, U_i, \varepsilon_{it})$, 
which when combined with $X_{it} = g_X(U_i, \varepsilon_{it})$ gives 
\eqref{eq::oneway_DGP}. However, the assumption does rule out dynamic regressors 
such as $X_{it} = Y_{i,t-1}$. The variance-weighted representation of the within 
estimator carries over to the dynamic case, but the interpretation is subtle, 
and we do not pursue this here.}
\begin{align}
    (Y_{it}, X_{it}) = g(U_i, \varepsilon_{it}),
    \label{eq::oneway_DGP}
\end{align}
for some unknown measurable function $g(\cdot,\cdot)$ and idiosyncratic shocks $\varepsilon_{it}$ (possibly multi-dimensional).
The structure parallels the cross-sectional case, with $U_i$ playing the role of the control variable $C_i$. For each unit $i$, we define the conditional linear projection coefficient\footnote{When $X_{it}$ is binary, $\beta_i$ reduces to the conditional average treatment effect $\beta_i = \mathbb{E}\left(Y_{it} \,\big|\, X_{it} = 1,\, U_i\right) - \mathbb{E}\left(Y_{it} \,\big|\, X_{it} = 0,\, U_i\right)$.}
\begin{align}
    \beta_i := \frac{{\rm Cov}\left(Y_{it},\, X_{it} \,\big|\, U_i\right)}{{\rm Var}\left(X_{it} \,\big|\, U_i\right)}.
    \label{eq::beta_i_oneway}
\end{align}
We assume covariance stationarity here, which implies that
$\beta_i$ does not depend on $t$. As before, we define the orthogonal decomposition
\begin{align}
    X_{it} = X_{it}^\parallel + X_{it}^\perp, \qquad X_{i}^\parallel := \mathbb{E}\left(X_{it} \,\big|\, U_i\right), \qquad X_{it}^\perp := X_{it} - X_{i}^\parallel.
    \label{eq::DefXperpPanel}
\end{align}
The pooled estimand is
\begin{align}
    \beta^* := \frac{{\rm Cov}\left(Y_{it},\, X_{it}^\perp\right)}{{\rm Var}\left(X_{it}^\perp\right)} = \frac{\mathbb{E}\left(Y_{it} X_{it}^\perp\right)}{\mathbb{E}\left[\left(X_{it}^\perp\right)^2\right]} = \mathbb{E}\left(w_i^* \, \beta_i\right),
    \label{eq::DefBetaStarPanel}
\end{align}
where the weights are given by
\begin{align}
    w_i^* := \frac{{\rm Var}\left(X_{it} \,\big|\, U_i\right)}{\mathbb{E}\left[{\rm Var}\left(X_{it} \,\big|\, U_i\right)\right]}.
    \label{eq::DefWeightsPanel}
\end{align}
Thus $\beta^*$ is a variance-weighted average of the unit-specific effects $\beta_i$, where as before the weights are non-negative and sum to one in expectation.

At first sight, it may seem difficult to estimate $\beta^*$ since the confounder $U_i$ is unobserved. However, the panel structure allows us to estimate $X_i^\parallel$ nonparametrically. Under covariance stationarity and mild mixing conditions on the serial dependence of $X_{it}$, we have
\begin{align}
    \frac{1}{T} \sum_{t=1}^T X_{it} \,  \to_P \,  X_i^\parallel ,
\end{align}
as $T \rightarrow \infty$.\footnote{%
The $T \to \infty$ framing in this subsection is used for expositional simplicity and for 
continuity with the two-way case in Section~\ref{section:estimators}, where it 
is genuinely needed. It is not essential for the variance-weighted interpretation 
of the within estimator in the one-way setting: for any fixed $T \geq 2$, the within demean gives 
$\mathbb{E}(Y_{it}\,\widehat{X_{it}^\perp})/\mathbb{E}((\widehat{X_{it}^\perp})^2) = \beta^*$ 
exactly, since the factor $(T-1)/T$ arising from demeaning cancels in the numerator 
and denominator.}
This suggests a natural estimator for $X_{it}^\perp$, namely
\begin{align}
    \widehat{X_{it}^\perp} := X_{it} - \frac{1}{T} \sum_{s=1}^T X_{is}.
\end{align}
The resulting estimator for $\beta^*$ is
\begin{align}
    \widehat{\beta} = \frac{\sum_{i=1}^n \sum_{t=1}^T Y_{it} \widehat{X_{it}^\perp}}{\sum_{i=1}^n \sum_{t=1}^T \left(\widehat{X_{it}^\perp}\right)^2}.
    \label{eq::FE_estimator}
\end{align}
This is the classical within group (WG) estimator, also known as the fixed effects (FE) estimator. Under appropriate regularity conditions, $\widehat{\beta} \to_P \beta^*$ as $n, T \to \infty$, see \citet{galvao2014estimation} and \citet{Juodis2020} for details.
Unlike the cross-sectional setting, no linearity condition analogous to \eqref{eq::LinearCondMean} is needed here. Covariance stationarity ensures that $X_i^\parallel$ is constant over $t$, so time-averaging removes it nonparametrically regardless of the functional form of $\mathbb{E}(X_{it} \,\big|\, U_i)$.

More generally, one may wish to estimate weighted averages of $\beta_i$ with user-specified weights. For each unit $i$, define the unit-specific estimator\footnote{Here and in \eqref{PesaranSmithConvergence}, we assume that $T^{-1} \sum_{t=1}^T \left(\widehat{X_{it}^\perp}\right)^2 \geq c > 0$ for all $i$. This rules out irregular identification issues that can arise for fixed $T$. See \citet{https://doi.org/10.3982/ECTA8220} for discussion.}
\begin{align}
    \widehat{\beta}_i := \frac{\sum_{t=1}^T Y_{it} \widehat{X_{it}^\perp}}{\sum_{t=1}^T \left(\widehat{X_{it}^\perp}\right)^2}.
\end{align}
The Mean Group estimator of \citet{Pesaran199579} uses equal weights,
\begin{align}
    \widehat{\beta}_{\rm MG} = \frac{1}{n} \sum_{i=1}^n \widehat{\beta}_i \to_P \mathbb{E}\left(\beta_i\right).
    \label{PesaranSmithConvergence}
\end{align}
For general weights $w_i$, one can form $\widehat{\beta}_w = \sum_{i=1}^n w_i \widehat{\beta}_i$. However, establishing consistency for such estimators (including other functionals of $\widehat{\beta}_i$) in the general model-free setup is more involved than for the pooled estimator $\widehat{\beta}$. See e.g. \citet{galvao2014estimation}, \citet{OkuiYanagi2015},
and \citet{jochmans2024inference}.

\section{Results for two-way unobserved heterogeneity}
\label{section:estimators}

\subsection{Setup and estimand}
\label{ssection::twoway_setup}

We now consider the setting where unobserved heterogeneity is present in both the cross-sectional and time-series dimensions. Consider a sample $(Y_{it}, X_{it})$ for $i = 1, \ldots, n$ and $t = 1, \ldots, T$. We assume that the observed variables $Z_{it} = (Y_{it}, X_{it})'$ are generated as\footnote{%
As in the one-way case, we allow for a DGP of the form 
$Y_{it} = h(X_{it}, U_i, V_t, \varepsilon_{it})$, which when combined with 
$X_{it} = g_X(U_i, V_t, \varepsilon_{it})$ gives \eqref{eq::AHK}. However, the 
assumption does rule out dynamic regressors such as $X_{it} = Y_{i,t-1}$.}
\begin{align}
    Z_{it} = g(U_i, V_t, \varepsilon_{it}),
    \label{eq::AHK}
\end{align}
where $g(\cdot,\cdot,\cdot)$ is a general measurable function, $U_i$ and $V_t$ are unobserved unit-specific and time-specific heterogeneity, and $\varepsilon_{it}$ are idiosyncratic shocks (all possibly multi-dimensional). This decomposition is closely related to the Aldous-Hoover representation for exchangeable arrays. See \citet{2019Menzel} for background and \citet{chiang2022standard} for an application to panel data.

The structure parallels Section~\ref{ssection::unobserved_oneway}, but now we condition on both $U_i$ and $V_t$. For each $(i,t)$, we define the conditional linear projection coefficient
\begin{align}
    \beta_{it} := \frac{{\rm Cov}\left(Y_{it},\, X_{it} \,\big|\, U_i,\, V_t\right)}{{\rm Var}\left(X_{it} \,\big|\, U_i,\, V_t\right)}.
    \label{eq::beta_it_twoway}
\end{align}
Due to the presence of time-specific heterogeneity $V_t$, this quantity now varies across both $i$ and $t$.\footnote{When $X_{it}$ is binary, $\beta_{it}$ reduces to the conditional average treatment effect $\beta_{it} = \mathbb{E}\left(Y_{it} \,\big|\, X_{it} = 1,\, U_i,\, V_t\right) - \mathbb{E}\left(Y_{it} \,\big|\, X_{it} = 0,\, U_i,\, V_t\right)$.} As before, we define the orthogonal decomposition
\begin{align}
    X_{it} = X_{it}^\parallel + X_{it}^\perp, \qquad X_{it}^\parallel := \mathbb{E}\left(X_{it} \,\big|\, U_i,\, V_t\right), \qquad X_{it}^\perp := X_{it} - X_{it}^\parallel,
    \label{eq::DefXperpTwoway}
\end{align}
and analogously for $Y_{it}$. Under identical distribution across $(i,t)$, the pooled estimand is
\begin{align}
    \beta^* := \frac{{\rm Cov}\left(Y_{it},\, X_{it}^\perp\right)}{{\rm Var}\left(X_{it}^\perp\right)} = \frac{\mathbb{E}\left(Y_{it} X_{it}^\perp\right)}{\mathbb{E}\left[\left(X_{it}^\perp\right)^2\right]} = \mathbb{E}\left(w_{it}^* \, \beta_{it}\right),
    \label{eq::DefBetaStarTwoway}
\end{align}
where the weights are given by
\begin{align}
    w_{it}^* := \frac{{\rm Var}\left(X_{it} \,\big|\, U_i,\, V_t\right)}{\mathbb{E}\left[{\rm Var}\left(X_{it} \,\big|\, U_i,\, V_t\right)\right]}.
    \label{eq::DefWeightsTwoway}
\end{align}
Thus $\beta^*$ is again a variance-weighted average of the conditional effects $\beta_{it}$, with non-negative weights that sum to one in expectation. Since $\mathbb{E}(Y_{it}^\parallel X_{it}^\perp) = 0$ by iterated expectations, we also have
\begin{align}
    \beta^* = \frac{{\rm Cov}\left(Y_{it}^\perp,\, X_{it}^\perp\right)}{{\rm Var}\left(X_{it}^\perp\right)} = \frac{\mathbb{E}\left(Y_{it}^\perp X_{it}^\perp\right)}{\mathbb{E}\left[\left(X_{it}^\perp\right)^2\right]}.
    \label{eq::DefBetaStarTwowayAlt}
\end{align}
This equivalent representation will be useful for understanding the IFE estimator in Section~\ref{ssection::PC_estimators}.
As in Section~\ref{ssection::unobserved_oneway}, the natural plug-in estimator for $\beta^*$ takes the form
\begin{align}
    \widehat{\beta} = \frac{\sum_{i=1}^n \sum_{t=1}^T Y_{it} \widehat{X_{it}^\perp}}{\sum_{i=1}^n \sum_{t=1}^T \left(\widehat{X_{it}^\perp}\right)^2},
    \label{eq::plugin_estimator}
\end{align}
where $\widehat{X_{it}^\perp}$ is an estimator of $X_{it}^\perp$. The key challenge is constructing $\widehat{X_{it}^\perp}$, which requires estimating $X_{it}^\parallel = \mathbb{E}\left(X_{it} \,\big|\, U_i,\, V_t\right)$.

This task is fundamentally different from the one-way setting of Section~\ref{ssection::unobserved_oneway}. In that setting, covariance stationarity implied $X_{it}^\parallel = X_i^\parallel$, which could be estimated nonparametrically by time-averaging. Here, $X_{it}^\parallel$ varies in both dimensions. Without further restrictions, it is impossible to distinguish $X_{it}^\parallel$ from $X_{it}^\perp$ using $X_{it}$ alone.

\subsection{Two-way fixed effects and its limitations}
\label{ssection::TWFE}

One approach to estimating $X_{it}^\perp$ is to impose an additive structure on $X_{it}^\parallel$,\footnote{%
Analogous to condition \eqref{eq::LinearCondMean}(b) in the cross-sectional setting, we could alternatively assume an additive structure on $Y_{it}^\parallel - \beta X_{it}^\parallel$. This would justify the same TWFE estimator under different assumptions.}
\begin{align}
    X_{it}^\parallel = g_1(U_i) + g_2(V_t),
\end{align}
for some measurable functions $g_1(\cdot)$ and $g_2(\cdot)$. Under this assumption, we can estimate $X_{it}^\perp$ by the two-way demeaning
\begin{align}
    \widehat{X_{it}^\perp} = X_{it} - \frac{1}{n} \sum_{j=1}^n X_{jt} - \frac{1}{T} \sum_{s=1}^T X_{is} + \frac{1}{nT} \sum_{j=1}^n \sum_{s=1}^T X_{js}.
\end{align}
The resulting estimator $\widehat{\beta}$ is the two-way fixed effects (TWFE) estimator. However, the additive structure is restrictive. When $X_{it}^\parallel$ contains terms that interact $U_i$ and $V_t$ non-additively, the two-way demeaning leaves behind a residual that is generally correlated with $X_{it}^\perp$, and TWFE fails to consistently estimate $\beta^*$.

The common correlated effects (CCE) approach of \citet{pesaran2006estimation}\footnote{See 
e.g.\ \citet{JUODIS2026106120} for a critical review of the recent literature on this 
methodology.} as well as generalizations of this approach (see e.g.\ \citealp{JuodisEtal2017} 
and \citealp{LU2026106183}) generally do not fully resolve the shortcomings of the TWFE 
transformation. In particular, the CCE approach builds upon the assumption that
\begin{align}
    X_{it}^\parallel = g_1(U_i)'g_2(V_t),
\end{align}
where now $g_1(\cdot)$ and $g_2(\cdot)$ are $R$-dimensional vectors. $R$ is generally assumed to be no larger than the number of regressors, e.g.\ $R = 1$ is 
the maximum permitted value for the single regressor case that we consider here.\footnote{This limitation can be circumvented by using additional 
proxy variables as in \citet{KarUrbWest2014} and \citet{JuodisSarafidis2019}, but generally 
the dimension $R$ remains finite as $n, T \to \infty$.} Hence, while the CCE approach does 
not impose additive structure on $X_{it}^\parallel$, it keeps the restriction that $X_{it}^\parallel$ needs to admit exact low-rank structure,
unlike the approaches that we describe below. Moreover, the CCE method imposes additional restrictions 
on the first conditional moments of $X_{it}^\parallel$ in the form of the so-called rank 
condition. In Appendix~\ref{ssection::appendix_TWFE_CCE} we provide a simple DGP 
that illustrates the potential failure of the TWFE and the CCE methods.

Next, we show that consistent estimation of $\beta^*$ is possible under 
weaker conditions. The key insight is that $X_{it}^\parallel = \mathbb{E}(X_{it} \,|\, U_i, V_t)$ 
can often be well-approximated by a low-rank structure, namely the leading principal components 
of $X_{it}$. This approximation is accurate when $X_{it}^\parallel$ is a sufficiently smooth 
function of $(U_i, V_t)$ and the dimensionality of $U_i$ and $V_t$ is not too high.

\subsection{Principal components based estimators}
\label{ssection::PC_estimators}

We now describe the two principal components based estimators that are the focus of this paper. Under suitable regularity conditions, both estimators consistently estimate the variance-weighted average effect $\beta^*$ defined in \eqref{eq::DefBetaStarTwoway}. Throughout, we write $Y$, $X$, $X^\parallel$, and $X^\perp$ for the $n \times T$ matrices with $(i,t)$ entries $Y_{it}$, $X_{it}$, $X_{it}^\parallel$, and $X_{it}^\perp$ respectively. %

The first estimator is the Principal Components (PC) estimator studied in \citet{GreenawayMcGrevy201248}. This is a two-step procedure. In the first step, we estimate $X_{it}^\parallel$ by the rank-$R$ approximation to $X$,
\begin{align}
    \widehat{X^\parallel} = \argmin_{\left\{G \in \mathbb{R}^{n \times T} : {\rm rank}(G) \leq R\right\}} \left\|X - G\right\|_2^2.
    \label{eq::PC_firststep}
\end{align}
This minimization has a closed-form solution given by the truncated singular value decomposition of $X$.\footnote{Let $X = \sum_{r=1}^{\min(n,T)} \sigma_r u_r v_r'$ be the singular value decomposition of $X$, with singular values $\sigma_1 \geq \sigma_2 \geq \ldots \geq 0$. By the Eckart-Young-Mirsky theorem, the solution to \eqref{eq::PC_firststep} is $\widehat{X^\parallel} = \sum_{r=1}^{R} \sigma_r u_r v_r'$.}
 Let $\widehat{X_{it}^\perp} := X_{it} - \widehat{X_{it}^\parallel}$ denote the residuals.
In the second step, the PC estimator is obtained by
\begin{align}
    \widehat{\beta}_{\rm PC} = \frac{\sum_{i=1}^n \sum_{t=1}^T Y_{it} \, \widehat{X_{it}^\perp}}{\sum_{i=1}^n \sum_{t=1}^T \left(\widehat{X_{it}^\perp}\right)^2}.
    \label{eq::PC_estimator}
\end{align}
Our second estimator is the Interactive Fixed Effects (IFE) estimator of \citet{bai2009panel}. This estimator estimates the coefficient $\beta$ and the factor structure jointly by solving
\begin{align}
    \widehat{\beta}_{\rm IFE} &= \argmin_{\beta \in \mathbb{R}} \, \min_{\left\{G \in \mathbb{R}^{n \times T} : {\rm rank}(G) \leq R\right\}} \left\|Y - X\beta - G\right\|_2^2 \notag \\
    &= \argmin_{\beta \in \mathbb{R}} \, \min_{(\lambda, f) \in \mathbb{R}^{nR + TR}} \sum_{i=1}^n \sum_{t=1}^T \left(Y_{it} - X_{it} \beta - \sum_{r=1}^R \lambda_{ir} f_{tr}\right)^2.
    \label{eq::IFE_estimator}
\end{align}
Here, the first line uses notation analogous to \eqref{eq::PC_firststep}, while the second line replaces the rank restriction on $G$ by the explicit low-rank representation $G_{it} = \lambda_i' f_t$, where $\lambda_i = (\lambda_{i1}, \ldots, \lambda_{iR})' \in \mathbb{R}^R$ are unit-specific factor loadings and $f_t = (f_{t1}, \ldots, f_{tR})' \in \mathbb{R}^R$ are time-specific factors. The second line also shows that $\widehat{\beta}_{\rm IFE}$ is simply obtained as the joint least squares estimator over $(\beta, \lambda, f)$.

The relationship between PC and IFE mirrors the two justifications for OLS in the cross-sectional setting. Recall from Section~\ref{ssection::cross_sectional} that the OLS estimator consistently estimates $\beta^*$ under either condition \eqref{eq::LinearCondMean}(a), which restricts the conditional mean of $X_i$, or condition \eqref{eq::LinearCondMean}(b), which restricts the conditional mean of $Y_i - \beta^* X_i$. By the Frisch-Waugh-Lovell theorem, the same OLS estimator is consistent for $\beta^*$ under either condition.
In the two-way panel setting, the analogous conditions are:
\begin{equation}
\label{eq::LowRankConditions}
\begin{aligned}
    &\text{(a)} \quad X^\parallel\text{\, is well-approximated by a rank-}R \text{ matrix}, \\
    &\text{(b)} \quad  Y^\parallel - \beta^* X^\parallel \text{\, is well-approximated by a rank-}R \text{ matrix}.
\end{aligned}
\end{equation}
The PC estimator $\widehat{\beta}_{\rm PC}$ corresponds to condition (a): it first removes the low-rank structure from $X$, then regresses $Y$ on the residuals. The IFE estimator $\widehat{\beta}_{\rm IFE}$ corresponds to condition (b): it jointly estimates $\beta$ and a low-rank approximation to $Y - X\beta$. Unlike the cross-sectional case, the two approaches are not numerically equivalent. However, under appropriate regularity conditions, both estimators consistently estimate $\beta^*$, as we show below.

Both estimators depend on a tuning parameter $R$, which controls the number of principal components or factors used. Our consistency results require $R$ to grow with $n$ and $T$.

\begin{remark}
The PC estimator $\widehat{\beta}_{\rm PC}$ only residualizes the regressor $X$ with 
respect to the low-rank structure, while leaving $Y$ unaffected. This is not the common 
approach in the literature; see e.g.\ \citet{westerlund2015cross}, where $Y$ is also 
residualized (either jointly with $X$, or independently of $X$). We expect our main 
conclusions to be robust to whether $Y$ is residualized or not, and do not discuss this 
alternative estimator theoretically.\footnote{We expect that similar results will also hold for the estimator studied in \citet{Beyhum02012023}.} The finite-sample performance of both PC variants 
is compared in Section~\ref{section:montecarlo}.
\end{remark}

\subsection{Consistency under high-level assumptions}
\label{ssection::consistency}

We now provide results for consistency of  the PC and IFE estimators
under high-level conditions.
In the following subsection, we provide lower-level sufficient conditions that are easier to interpret and verify.
We begin with the PC estimator. Recall that this estimator 
is defined by \eqref{eq::PC_firststep} and \eqref{eq::PC_estimator}
as a function of $Y$, $X$, and $R$. We now write
$R = R_{nT}$, and allow it to depend on the sample size.

\begin{theorem}[\bf Consistency of PC estimator]
\label{thm::PC_consistencyNEW}
  Let  $\beta^{*} \in \mathbb{R}$ be a constant
      and  $X^\parallel \in \mathbb{R}^{n \times T}$ 
      be a random matrix, and define  $E:=Y - \beta^* \, X$
      and $X^\perp:=X- X^\parallel$.
     Assume that, as $n,T \rightarrow \infty$
     we have
    \begin{enumerate}[(i)]

      \item  
$\frac 1 {nT} \sum_{i=1}^n \sum_{t=1}^T  X^\perp_{it} \, E_{it}  =o_{P}(1)$,
and 
$\frac 1 {nT} \sum_{i=1}^n \sum_{t=1}^T    X_{it} ^2  ={\cal O}_{P}(1)$,
and  $\frac 1 {nT} \sum_{i=1}^n \sum_{t=1}^T    Y_{it} ^2  ={\cal O}_{P}(1)$.

   \item 
   $\frac 1 {nT} \sum_{i=1}^n \sum_{t=1}^T   \left( X^\perp_{it} \right)^2  \to_P  c$, for a constant $c>0$.

       \item 
        $\frac{R_{nT}} {nT} \| X^\perp \|^2 = o_P(1)$.

       \item $\min_{\big\{ G \in \mathbb{R}^{n \times T}: {\rm rank}(G) \leq R_{nT} \big\}} 
   \frac 1 {nT} \left\|  X^\parallel -  G \right\|^2_2 
  = o_P(1) $.
   \end{enumerate}
    Then,
    $$
        \widehat \beta_{\rm PC} =  \beta^* + o_P(1).
    $$
\end{theorem}

The theorem is stated in a self-contained way: it assumes the existence of a constant $\beta^*$, a matrix $X^\parallel$, and derived quantities $E = Y - \beta^* X$ and $X^\perp = X - X^\parallel$ satisfying conditions (i)--(iv), and concludes that $\widehat{\beta}_{\rm PC}$ consistently estimates $\beta^*$. No specific interpretation of these quantities is assumed. Of course, the theorem becomes useful when we set $X_{it}^\parallel := \mathbb{E}(X_{it} \,|\, U_i, V_t)$ and define $\beta^*$ as the variance-weighted average in \eqref{eq::DefBetaStarTwoway}. This is what we do in the sufficient conditions below.

We now briefly discuss each condition. Condition (i) requires that $X^\perp$ and $E$ are asymptotically uncorrelated, and that $X$ and $Y$ have bounded second moments. Condition (ii) is a non-degeneracy condition ensuring that $X_{it}^\perp$ has positive variance in the limit. Condition (iii) requires that the spectral norm of $X^\perp$ does not grow too fast relative to $nT/R_{nT}$. Condition (iv) is the key low-rank assumption: it requires that $X^\parallel$ can be well-approximated by a rank-$R_{nT}$ matrix. This corresponds to condition \eqref{eq::LowRankConditions}(a) in the previous subsection. Note that (iii) requires $R_{nT}$ to not be too large, while (iv) requires $R_{nT}$ to not be too small. The sufficient conditions below assume that $R_{nT}$ grows with $n$ and $T$, but slowly.

We now state the analogous result for the IFE estimator.  

\begin{theorem}[\bf Consistency of IFE estimator]
\label{thm::IFE_consistencyNEW}
 Let  $\beta^{*} \in \mathbb{R}$ be a constant, and define  $E:=Y - \beta^* \, X$
      and $E^\perp:=E- E^\parallel$, where $E^\parallel \in \mathbb{R}^{n \times T}$ 
      is a random matrix.
     Assume that, as $n,T \rightarrow \infty$ we have

    \begin{enumerate}[(i)]

     \item    
     $\frac 1 {nT} \sum_{i=1}^n \sum_{t=1}^T  X_{it} \, E^\perp_{it}  =o_{P}(1)$,
     and     
     $\frac 1 {nT} \sum_{i=1}^n \sum_{t=1}^T    X_{it} ^2  ={\cal O}_{P}(1)$.
       
       \item 
        $ 
     \min_{\big\{ G \in \mathbb{R}^{n \times T}: {\rm rank}(G) \leq 2R_{nT} \big\}} 
     \frac{1}{nT}\left\|  X -  G \right\|^2_2 
  \geq c$, wpa1, for a constant $c>0$.

       \item $\frac{R_{nT}} {nT} \| E^\perp \|^2 = o_P(1)$.
    
       \item $\min_{\big\{ G \in \mathbb{R}^{n \times T}: {\rm rank}(G) \leq R_{nT} \big\}} 
   \frac{1} {nT} \left\|  E^\parallel -  G \right\|^2_2 
  = o_P(1) $.

    \end{enumerate}
    Then,
    $$
        \widehat \beta_{\rm IFE} =  \beta^* + o_P(1).
    $$

\end{theorem}

As with Theorem~\ref{thm::PC_consistencyNEW}, this result is stated in a self-contained way: it assumes the existence of a constant $\beta^*$, a matrix $E^\parallel$, 
and derived quantities $E = Y - \beta^* X$ and $E^\perp = E - E^\parallel$ satisfying conditions (i)--(iv). The theorem becomes useful when we set $E_{it}^\parallel := \mathbb{E}(E_{it} \,|\, U_i, V_t)$ and define $\beta^*$ as the variance-weighted average in \eqref{eq::DefBetaStarTwoway}.

Compared to Theorem~\ref{thm::PC_consistencyNEW}, the roles of $X$ and $E$ are reversed, reflecting the different structure of the IFE estimator. Condition (i) requires that $X$ is asymptotically uncorrelated with the idiosyncratic component $E^\perp$, and that $X$ has bounded second moments. Condition (ii) is a non-collinearity condition requiring that $X$ cannot be well-approximated by a low-rank matrix.\footnote{Here, similarly to \citet{FreemanWeidner2021}, we need to assume that ${\rm rank}(G) \leq 2R_{nT}$ and not simply ${\rm rank}(G) \leq R_{nT}$.} This ensures that the coefficient $\beta$ can be identified separately from the factor structure. Condition (iii) requires that the spectral norm of $E^\perp$ does not grow too fast relative to $nT/R_{nT}$. Condition (iv) is the key low-rank assumption: it requires that $E^\parallel$ can be well-approximated by a rank-$R_{nT}$ matrix. This corresponds to condition \eqref{eq::LowRankConditions}(b) in the previous subsection. As for the PC estimator, (iii) requires $R_{nT}$ to not be too large, while (iv) requires $R_{nT}$ to not be too small.

The proofs of both theorems are given in Appendix~\ref{section_appendix:proofs}. The proof strategy builds on the approach of \citet{FreemanWeidner2021}, adapted to our setting.

\subsection{Consistency under low-level assumptions}
\label{ssection::lowlevel}

We now provide primitive sufficient conditions on the data generating process \eqref{eq::AHK} that imply the high-level conditions in Theorems~\ref{thm::PC_consistencyNEW} and \ref{thm::IFE_consistencyNEW}. We emphasize that these are only sufficient conditions: the high-level conditions may also be satisfied under alternative assumptions not covered by the theorem below.

\begin{theorem}[\bf Consistency under low-level assumptions]
\label{thm::lowlevel}
Let $Z_{it} = (Y_{it}, X_{it})' \in \mathbb{R}^2$ be generated as $Z_{it} = g(U_i, V_t, \varepsilon_{it})$ for some measurable function $g : \mathcal{U} \times \mathcal{V} \times \mathcal{E} \to \mathbb{R}^2$. Define $X_{it}^\parallel := \mathbb{E}(X_{it} \,|\, U_i, V_t)$, $X_{it}^\perp := X_{it} - X_{it}^\parallel$, and $\beta^*$ as in \eqref{eq::DefBetaStarTwoway}. Assume:
\begin{enumerate}[(i)]
    \item $(U_i : i \geq 1)$ are i.i.d., and $(V_t : t \geq 1)$ is strictly stationary and $\alpha$-mixing with mixing coefficients $\alpha_V(\tau)$ satisfying $\sum_{\tau=1}^{\infty}\alpha_V(\tau) < \infty$.
    
    \item $(\varepsilon_{it}:i \geq 1,t \geq 1)$ for $\varepsilon_{it} \in \mathbb{R}^{d_\varepsilon}$ is independent of $(U, V) := ((U_i)_{i \geq 1}, (V_t)_{t \geq 1})$, for fixed integers $d_\varepsilon$. Time series $(\varepsilon_{it} : t \geq 1)$ are i.i.d. across $i$, and are strictly stationary and $\alpha$-mixing with mixing coefficients $\alpha_\varepsilon(\tau)$ satisfying $\sum_{\tau=1}^{\infty}\alpha_\varepsilon(\tau)<\infty$.

    \item $\sup_{u \in \mathcal{U}, v \in \mathcal{V}} \mathbb{E}[\|g(u, v, \varepsilon_{it})\|^4] < \infty$, and ${\rm Var}(X_{it} \,|\, U_i, V_t) > 0$ almost surely.
    
    \item $U_i \in \mathbb{R}^{d_U}$ and $V_t \in \mathbb{R}^{d_V}$ for fixed integers $d_U, d_V \geq 1$. The conditional mean functions $h_X(u,v) := \mathbb{E}(X_{it} \,|\, U_i = u, V_t = v)$ and $h_E(u,v) := \mathbb{E}(Y_{it} - \beta^* X_{it} \,|\, U_i = u, V_t = v)$ are $s$-times continuously differentiable with $s > \max(d_U, d_V)/2$.
    
    \item $R_{nT} \to \infty$ and $R_{nT} = o\left(\min\left(\sqrt{T}, \sqrt{n} / (\log T)^2\right)\right)$.
\end{enumerate}
Then the assumptions of Theorems~\ref{thm::PC_consistencyNEW} and \ref{thm::IFE_consistencyNEW} are satisfied, and hence $\widehat{\beta}_{\rm PC} = \beta^* + o_P(1)$ and $\widehat{\beta}_{\rm IFE} = \beta^* + o_P(1)$.
\end{theorem}
The proof of Theorem~\ref{thm::lowlevel} is given in Appendix~\ref{section_appendix:proofs}.
We now provide an informal discussion of why these conditions imply the high-level assumptions of Theorems~\ref{thm::PC_consistencyNEW} and \ref{thm::IFE_consistencyNEW}. The formal proof proceeds through an intermediate set of conditions stated in Lemma~\ref{lemma::medium} in the appendix.

Conditions (i) and (ii) specify the dependence structure of the data generating process. Units are drawn independently, while time periods may exhibit serial dependence through both the common shocks $V_t$ and the idiosyncratic shocks $\varepsilon_{it}$. The $\alpha$-mixing assumption with summable coefficients ensures that this serial dependence decays sufficiently fast. These conditions are standard in the panel data literature with common factors; see \citet{fernandez_weidner_2014} and \citet{SuEtAL2014} for related discussions. Together with condition (iii), they imply laws of large numbers and bounds on spectral norms that are needed to verify the high-level conditions. In particular, the independence of $\varepsilon_{it}$ from $(U, V)$ ensures that $X^\perp$ and $E^\perp$ are asymptotically uncorrelated, which is essential for conditions (i) in both Theorems~\ref{thm::PC_consistencyNEW} and \ref{thm::IFE_consistencyNEW}. The mixing conditions also yield bounds on the spectral norms $\|X^\perp\|$ and $\|E^\perp\|$ that, combined with condition (v), imply conditions (iii) in both theorems. See \citet{wang2022lowrankpanelquantileregression} for related spectral norm bounds.

Condition (iii) provides moment bounds and a non-degeneracy requirement. The uniform fourth moment bound ensures that sample averages converge to their population counterparts. The requirement ${\rm Var}(X_{it} \,|\, U_i, V_t) > 0$ ensures that $X_{it}^\perp$ has positive variance, which is needed for condition (ii) in Theorem~\ref{thm::PC_consistencyNEW} and condition (ii) in Theorem~\ref{thm::IFE_consistencyNEW}.

Condition (iv) is the key smoothness assumption. It requires that the conditional mean functions $h_X$ and $h_E$ are sufficiently smooth relative to the underlying dimensionality of the unobserved heterogeneity. Classical results in approximation theory show that for an $s$-smooth function on $\mathbb{R}^d$, the singular values of the associated integral operator decay as $\sigma_r = \largeO(r^{-s/d})$. When $s > d/2$, this decay is fast enough to ensure that $\sum_{r > R} \sigma_r^2 \to 0$ as $R \to \infty$. Applied to our setting with $d = \max(d_U, d_V)$, this implies that $X^\parallel$ and $E^\parallel$ can be well-approximated by low-rank matrices, which is precisely what conditions (iv) in Theorems~\ref{thm::PC_consistencyNEW} and \ref{thm::IFE_consistencyNEW} require.

Condition (v) restricts the growth rate of $R_{nT}$. The requirement $R_{nT} \to \infty$ ensures that the low-rank approximation error vanishes asymptotically. The upper bound $R_{nT} = o(\min(\sqrt{T}, \sqrt{n}/(\log T)^2))$ ensures that $R_{nT}$ does not grow so fast as to violate the spectral norm conditions (iii) in the high-level theorems.

Theorem~\ref{thm::lowlevel} is the main theoretical contribution of this paper. It establishes that the PC and IFE estimators have well-defined probability limits under fully nonparametric assumptions on the data generating process. Existing results for these estimators, such as \citet{bai2009panel}, typically assume a parametric model of the form $Y_{it} = X_{it} \beta + \lambda_i' f_t + \varepsilon_{it}$. In contrast, our conditions do not impose any parametric structure. The representation $Y_{it} = X_{it}\beta^*  + G_{it} + \xi_{it}$ is not assumed but derived from the definition of $\beta^*$ as a variance-weighted average. The low-rank condition (iv) restricts only the conditional mean functions $h_X$ and $h_E$, not the full distribution of $(Y_{it}, X_{it})$ given $(U_i, V_t)$. The key requirement for consistency, beyond regularity conditions, is that the number of factors $R_{nT}$ grows with $n$ and $T$, but slowly. This allows the low-rank approximation to capture increasingly complex patterns in the conditional means while ensuring that estimation error remains controlled.

\subsection{Results conditional on $U$ and $V$}
\label{ssection::conditional}

The results in the previous subsections define the estimand $\beta^*$ as a population 
quantity that remains fixed as $n, T \to \infty$. This requires distributional assumptions 
on $(U_i)$ and $(V_t)$, such as the i.i.d.\ and stationarity conditions in 
Theorem~\ref{thm::lowlevel}. An alternative approach, more in line with the standard 
interactive fixed effects literature, is to condition on $(U, V) = ((U_i)_{i=1}^n, 
(V_t)_{t=1}^T)$ throughout. This allows for arbitrary sequences $(U_i)$ and $(V_t)$, 
at the cost of working with an estimand that depends on $n$ and $T$. For fixed sequences 
$(U_1, \ldots, U_n)$ and $(V_1, \ldots, V_T)$, we define the conditional (or finite-population) estimand
\begin{align}
    \beta^*_{nT} := \frac{\sum_{i=1}^n \sum_{t=1}^T {\rm Cov}(Y_{it},\, X_{it} \,|\, U_i,\, V_t)}{\sum_{i=1}^n \sum_{t=1}^T {\rm Var}(X_{it} \,|\, U_i,\, V_t)} = \sum_{i=1}^n \sum_{t=1}^T w_{it}^{nT} \, \beta_{it},
    \label{eq::DefBetaStarNT}
\end{align}
where $\beta_{it}$ is defined in \eqref{eq::beta_it_twoway} and the weights are
\begin{align}
    w_{it}^{nT} := \frac{{\rm Var}(X_{it} \,|\, U_i,\, V_t)}{\sum_{j=1}^n \sum_{s=1}^T {\rm Var}(X_{js} \,|\, U_j,\, V_s)}.
    \label{eq::DefWeightsNT}
\end{align}
Thus $\beta^*_{nT}$ is a variance-weighted average of the conditional effects $\beta_{it}$, 
with non-negative weights satisfying $\sum_{i=1}^n \sum_{t=1}^T w_{it}^{nT} = 1$. Unlike $\beta^*$ in 
\eqref{eq::DefBetaStarTwoway}, the estimand $\beta^*_{nT}$ depends on the realized sequences 
$(U_i)$ and $(V_t)$ and may change with $n$ and $T$, but its interpretation is unchanged.

Theorems~\ref{thm::PC_consistencyNEW} and \ref{thm::IFE_consistencyNEW} apply unchanged 
when $\beta^*$ is replaced by $\beta^*_{nT}$ and all statements are interpreted conditionally 
on $(U, V)$. The proofs require no modification: the key orthogonality condition 
$\sum_{i=1}^n \sum_{t=1}^T \mathbb{E}[X_{it}^\perp E_{it} \,|\, U_i, V_t] = 0$ holds by construction of 
$\beta^*_{nT}$, and all other conditions concern matrices and spectral norms that are 
well-defined conditional on $(U, V)$.

The following theorem provides sufficient conditions for consistency in the conditional 
setting. It parallels Theorem~\ref{thm::lowlevel}, but replaces the distributional 
assumptions on $(U_i)$ and $(V_t)$ with direct conditions on the low-rank approximability 
of the conditional mean matrices.

\begin{theorem}[\bf Conditional consistency]
\label{thm::conditional}
Let $Z_{it} = (Y_{it}, X_{it})' \in \mathbb{R}^2$ be generated as $Z_{it} = g(U_i, V_t, 
\varepsilon_{it})$ for fixed sequences $(U_1, \ldots, U_n)$ and $(V_1, \ldots, V_T)$, and 
some measurable function $g : \mathcal{U} \times \mathcal{V} \times \mathcal{E} \to 
\mathbb{R}^2$. Define $X_{it}^\parallel := \mathbb{E}(X_{it} \,|\, U_i, V_t)$, 
$X_{it}^\perp := X_{it} - X_{it}^\parallel$, and $\beta^*_{nT}$ as in 
\eqref{eq::DefBetaStarNT}. Assume:
\begin{enumerate}[(i)]
    \item 
    $(\varepsilon_{it}:i \geq 1,t \geq 1)$ for $\varepsilon_{it} \in \mathbb{R}^{d_\varepsilon}$ for fixed integers $d_\varepsilon$. Time series $(\varepsilon_{it} : t \geq 1)$ are independent across $i$, and are $\alpha$-mixing with mixing coefficients $\alpha_{i}(\tau)$ 
    satisfying $\sup_{i \leq n}\alpha_{i}(\tau)\leq \alpha(\tau)$ with $\sum_{\tau=1}^\infty \alpha(\tau) < \infty$.
    \item $\sup_{i \leq n,\, t \leq T} \mathbb{E}[\|g(U_i, V_t, \varepsilon_{it})\|^4] < \infty$, 
    and $\frac{1}{nT}\sum_{i,t}{\rm Var}(X_{it} \,|\, U_i, V_t) \geq c > 0$.
    \item The matrices $X^\parallel$ and $E^\parallel$ satisfy condition~(iv) of 
    Theorems~\ref{thm::PC_consistencyNEW} and \ref{thm::IFE_consistencyNEW} respectively.
    \item $R_{nT} \to \infty$ and $R_{nT} = o\left(\min\left(\sqrt{T}, \sqrt{n} / 
    (\log T)^2\right)\right)$.
\end{enumerate}
Then $\widehat{\beta}_{\rm PC} = \beta^*_{nT} + o_P(1)$ and $\widehat{\beta}_{\rm IFE} = 
\beta^*_{nT} + o_P(1)$, where the $o_P(1)$ term is conditional on $(U, V)$.
\end{theorem}

The proof of Theorem~\ref{thm::conditional} is given in Appendix~\ref{section_appendix:proofs}.
Condition~(iii) of Theorem~\ref{thm::conditional} assumes directly that the conditional 
mean matrices admit good low-rank approximations. In Theorem~\ref{thm::lowlevel}, this 
condition was derived from smoothness of the conditional mean functions $h_X$ and $h_E$, 
combined with the i.i.d.\ and stationarity assumptions on $(U_i)$ and $(V_t)$. In the 
conditional setting, smoothness alone does not suffice: we also require that the realized 
sequences $(U_i)$ and $(V_t)$ are sufficiently spread out, in the sense that the singular 
functions of $h_X$ and $h_E$ do not concentrate on these sequences. In 
Appendix~\ref{appendix::conditional_U_V} we discuss a sufficient condition for 
Condition~(iii) of Theorem~\ref{thm::conditional}.

\begin{remark}[\bf Relationship between $\beta^*_{nT}$ and $\beta^*$]
\label{remark::betaNT_vs_beta}
Under the conditions of Theorem~\ref{thm::lowlevel}, the two estimands coincide 
asymptotically: $\beta^*_{nT} \to_P \beta^*$ as $n, T \to \infty$. This follows from 
the law of large numbers. The numerator satisfies
\begin{equation}
\frac{1}{nT} \sum_{i,t} {\rm Cov}(Y_{it}, X_{it} \,|\, U_i, V_t) \to_P 
\mathbb{E}\left[{\rm Cov}(Y_{it}, X_{it} \,|\, U_i, V_t)\right] = {\rm Cov}(Y_{it}, X_{it}^\perp),
\end{equation}
and similarly the denominator converges to ${\rm Var}(X_{it}^\perp)$. Thus when the 
distributional assumptions of Theorem~\ref{thm::lowlevel} hold, the conditional and 
unconditional formulations yield the same limit. The conditional formulation is more 
general in that it does not require these distributional assumptions, but the estimand 
$\beta^*_{nT}$ then depends on the particular sequences $(U_i)$ and $(V_t)$ observed 
in the sample.
\end{remark}

\section{Discussion and Extensions}
\label{section:discussion}

This section discusses limitations of the results presented so far as well as natural 
extensions of the framework. None of the discussions below involve new formal results, our goal here is just to clarify the scope of the paper and to point to directions 
for future work.

\subsection{Inference}
\label{ssection:inference}

The consistency results in Theorems~\ref{thm::PC_consistencyNEW} 
and~\ref{thm::IFE_consistencyNEW} are derived under minimal restrictions on the 
underlying DGP and are silent on the distribution of the two estimators. More progress 
can be achieved if additional restrictions on the DGP and on the approximation rate of 
$R_{nT}$ are imposed. In Appendix~\ref{ssection::freeman_weidner} we follow 
\citet{FreemanWeidner2021} and derive rate results for the IFE estimator. As in that 
paper, the rate of convergence is driven by the approximation bias associated with 
$R_{nT} \to \infty$. The fact that the asymptotic distribution is dominated by this 
bias term is further confirmed by our Monte Carlo study in 
Section~\ref{section:montecarlo}. Since no general nonparametric method for 
correcting this bias is available, the PC and IFE estimators cannot be directly used 
for inference on $\beta^*$.\footnote{Even in settings where $E_{it}^\parallel$ or 
$X_{it}^\parallel$ have an exact finite factor structure, the estimators we consider suffer from 
the incidental parameters problem (see e.g.\ \citealp{westerlund2015cross} and 
\citealp{MoonWeidnerET}) and require bias correction.}

Some progress is possible if further restrictions on the DGP are imposed. Following 
\citet{https://doi.org/10.3982/ECTA15238}, \citet{FreemanWeidner2021}, and 
\citet{beyhum2024inferencediscretizingunobservedheterogeneity}, one approach is to 
discretize the unobserved heterogeneity. We take the suggestion of 
\citet{beyhum2024inferencediscretizingunobservedheterogeneity} as a leading example. 
The procedure assumes that there exist injective functions $a_i$ and $b_t$ of $U_i$ 
and $V_t$ respectively, with finite-sample counterparts $\widehat{a}_i$ and 
$\widehat{b}_t$, and proceeds in three steps:
\begin{description}
    \item[Step 1.] Cluster units $i = 1, \ldots, n$ into $G$ clusters based on 
    $\widehat{a}_i$. Denote
their cluster labels by $g_{1},\ldots,g_{n}\in \{1,\ldots,G\}$.
    \item[Step 2.] Cluster time periods $t = 1, \ldots, T$ into $C$ clusters based on 
    $\widehat{b}_t$.  Denote
their cluster labels by $c_{1},\ldots,c_{T}\in \{1,\ldots,C\}$.
    \item[Step 3.] Apply the pooled OLS estimator to $\widehat{Z}_{it} := Z_{it} - 
    \overline{Z}_{g_i t} - \overline{Z}_{i c_t} + \overline{Z}_{g_i c_t}$.
\end{description}
Following \citet{beyhum2024inferencediscretizingunobservedheterogeneity}, we refer to 
this as the two-way grouped fixed effects (TWGFE) estimator. Their Monte Carlo results 
indicate that its finite-sample distribution is well approximated by a normal 
distribution that ignores estimation uncertainty from the grouping step, provided the 
dimension of $U_i$ and $V_t$ is small (mirroring assumption (iv) of 
Theorem~\ref{thm::lowlevel}) and K-means clustering is used.

The TWGFE approach is not yet theoretically justified in our setting, due to the 
estimation of high-dimensional nuisance parameters indexed by group membership $g_i$ 
and $c_t$.\footnote{We conjecture that consistency of TWGFE for $\beta^*$ can be 
established under the low-level conditions of Theorem~\ref{thm::lowlevel} combined 
with the injectivity conditions in 
\citet{beyhum2024inferencediscretizingunobservedheterogeneity}.} The theoretical 
results in \citet{beyhum2024inferencediscretizingunobservedheterogeneity} are instead 
established via sample splitting, as it is common in the double machine learning literature 
\citep{10.1111/ectj.12097}. Such sample splitting could be justified in our setting 
only if the projection errors $E_{it}^\perp$ and $X_{it}^\perp$ were i.i.d.\ over 
both dimensions. Since in our setting these errors have no structural meaning beyond 
population linear projection residuals, imposing such restrictions is not appropriate.

A further difficulty is that the convergence rate of the infeasible pooled OLS estimator 
(even for $E^\parallel$ or $X^\parallel$ known) can range from $\sqrt{\min(n,T)}$ to 
$\sqrt{nT}$, depending on the DGP, without any intermediate rate being ruled out. As 
a result, uniform inference on $\beta^*$ is generally not possible in our setting, as 
discussed in \citet{2019Menzel} and \citet{Juodis2020}.

\subsection{Estimation of targeted treatment effects}
\label{ssection:targeted}

The pooled estimand $\beta^*$ in \eqref{eq::DefBetaStarTwoway} uses variance-driven 
weights $w_{it}^*$ that are determined by the data generating process rather than the 
researcher. In many applications one may prefer a user-specified weighting scheme. A 
natural target is $\beta_w := \mathbb{E}(w_{it}\,\beta_{it})$ for some 
chosen weights $w_{it}$, with the equal-weighted case $w_{it} = 1$ being the panel analogue 
of the mean group estimator of \citet{Pesaran199579}. Consistently estimating $\beta_w$ 
requires a way to estimate the unit-time specific coefficients $\beta_{it}$ themselves.
The standard starting point is to write
\begin{equation}
    Y_{it} = \beta_{it}\,X_{it} + E_{it}, \qquad 
    E_{it} = E_{it}^\parallel + E_{it}^\perp,
    \label{eq::hetmodel}
\end{equation}
and impose structure on how $\beta_{it}$ varies with $(i,t)$. Three 
specifications proposed in the literature are:
\begin{enumerate}[(i)]
    \item $\beta_{it} = \beta_i$, as in the mean group estimator of \citet{Pesaran199579};
    \item $\beta_{it} = \beta_i + \beta_t$, as in \citet{KeaneNeal2020} and \citet{LU2023694};
    \item $\beta_{it} = g_1(U_i)'g_2(V_t)$, an exact finite factor structure, 
          as in \citet{chernozhukov2019inference}.
\end{enumerate}
All three approaches reduce estimation of $\beta_{it}$ to a finite-dimensional problem 
by imposing separability (either additive or multiplicative) in the $(U_i, V_t)$ 
arguments. Estimation in each case further requires that $E_{it}^\parallel$ has an exact 
low-rank structure. These restrictions are convenient but potentially strong.\footnote{For 
example, the simple specification
$\beta_{it} = \beta(U_i, V_t) = (g_{11}(U_i) + g_{21}(V_t))/(g_{12}(U_i) + g_{22}(V_t))$
satisfies neither additive nor multiplicative separability, and is not covered by any 
of the three approaches above.}

An alternative route for estimating $\beta_w = \mathbb{E}(w_{it}\,\beta_{it})$ is as 
follows. Let $Z_{it} = (Y_{it}, X_{it})'$ and define $Z_{it}^\perp = Z_{it} - 
\mathbb{E}(Z_{it} \,|\, U_i, V_t)$ as before. Then $\beta_{it}$ is a ratio of elements 
of the $2 \times 2$ conditional variance matrix
\begin{equation}
    \Sigma_{it} := \mathbb{E}\bigl(Z_{it}^\perp(Z_{it}^\perp)' \,\big|\, U_i,\, V_t\bigr) ,
\end{equation}
namely $\beta_{it} = \Sigma_{it,YX}/\Sigma_{it,XX}$, where $\Sigma_{it,YX}$ and 
$\Sigma_{it,XX}$ denote the $(Y,X)$ and $(X,X)$ elements respectively. Now observe that
\begin{equation}
    Z_{it}^\perp(Z_{it}^\perp)' = \Sigma_{it} + \Xi_{it}, \qquad 
    \mathbb{E}(\Xi_{it} \,|\, U_i, V_t) = 0,
    \label{eq::Sigma_decomp}
\end{equation}
so each element of $Z_{it}^\perp(Z_{it}^\perp)'$ has conditional mean equal to the 
corresponding element of $\Sigma_{it}$, with a mean-zero noise term $\Xi_{it}$. This is 
the same approximate low-rank structure exploited in the main results, now applied to 
the $2\times 2$ outer product matrix rather than to $X$ or $E$ alone. Using the identity
\begin{equation}
    \Sigma_{it} = \mathbb{E}\bigl(Z_{it}(Z_{it})' \,\big|\, U_i, V_t\bigr) 
    - \mathbb{E}\bigl(Z_{it} \,\big|\, U_i, V_t\bigr)
      \mathbb{E}\bigl(Z_{it} \,\big|\, U_i, V_t\bigr)',
    \label{eq::Sigma_expand}
\end{equation}
one can apply a low-rank approximation separately to each of the $n\times T$ matrices 
with $(i,t)$ entries $Y_{it}X_{it}$, $X_{it}^2$, $Y_{it}$, and $X_{it}$, obtaining 
estimators of the corresponding conditional means, and then form 
$\widehat{\Sigma}_{it,YX}$ and $\widehat{\Sigma}_{it,XX}$ by plugging in. Under 
smoothness conditions analogous to those of Theorem~\ref{thm::lowlevel}, this yields 
consistent estimators of the elements of $\Sigma_{it}$. This leads to the estimator
\begin{equation}
    \widehat{\beta}_w = \frac{1}{nT}\sum_{i=1}^n\sum_{t=1}^T 
    w_{it}\,\frac{\widehat{\Sigma}_{it,YX}}{\widehat{\Sigma}_{it,XX}}.
    \label{eq::betaw_estimator}
\end{equation}
We stress that \eqref{eq::betaw_estimator} is intended as a conceptual proposal rather 
than a final estimator. In particular, the denominator $\widehat{\Sigma}_{it,XX}$ 
estimates a conditional variance and may be close to zero for some $(i,t)$, so some 
form of regularization (such as thresholding) is needed to keep the ratio well-behaved. 
We leave a formal treatment of these questions for future work.

The estimator in \eqref{eq::betaw_estimator} can be seen as a natural extension of 
standard mean group methods to the two-way setting. In the one-way case, 
$\widehat{\Sigma}_{it,YX}$ and $\widehat{\Sigma}_{it,XX}$ are only allowed to vary 
along a single dimension (either $i$ or $t$), so that the conditional expectations 
can be estimated by simple within-group averages. Here the conditioning set is 
bivariate $(U_i, V_t)$, and the low-rank approximation plays the role that 
within-group averaging plays in the one-way case, handling both dimensions jointly. 
As an alternative to principal components, the discretization approach of 
\citet{beyhum2024inferencediscretizingunobservedheterogeneity} can also be used: 
after clustering units into groups $g_i$ and time periods into groups $c_t$ as 
described in Section~\ref{ssection:inference}, one can estimate $\widehat{\beta}_{g_i c_t}$ 
for every cell $(g_{i},c_{t})$ of the resulting partition and then form the weighted average 
$\widehat{\beta}_w$ directly.

Finally, our discussion so far addressed consistent estimation of $\beta_w$ 
but was silent on inference. Inference for this class of estimands is generally 
more involved than for the pooled estimand $\beta^*$, and typically requires either 
additional non-degeneracy conditions on the distribution of $\beta_{it}$ 
(see e.g.\ \citealp{KeaneNeal2020} and \citealp{LU2023694}) or sample splitting and regularization schemes 
(see e.g.\ \citealp{chernozhukov2019inference}).

\subsection{Additional regressors}
\label{ssection::additional_regressors}

So far we have taken $X_{it}$ and $\beta$ to be scalars. We now consider the 
extension to $X_{it} \in \mathbb{R}^K$ and $\beta \in \mathbb{R}^K$, with $Y_{it} \in \mathbb{R}$ 
unchanged. Define $X_{it}^\parallel := \mathbb{E}(X_{it} \,|\, U_i, V_t) \in \mathbb{R}^K$ 
and $X_{it}^\perp := X_{it} - X_{it}^\parallel \in \mathbb{R}^K$ as before. 
The conditional projection coefficient \eqref{eq::beta_it_twoway} then also becomes a $K$-vector, 
\begin{align}
    \beta_{it} := 
    \Bigl(\mathbb{E}\bigl(X_{it}^\perp (X_{it}^\perp)' \,\big|\, U_i, V_t\bigr)\Bigr)^{-1}
    \mathbb{E}\bigl(X_{it}^\perp Y_{it} \,\big|\, U_i, V_t\bigr)
    .
    \label{eq::beta_it_twoway_vec}
\end{align}
The pooled estimand \eqref{eq::DefBetaStarTwoway} generalises to
\begin{align}
    \beta^* 
    &:=
    \Bigl(\mathbb{E}\bigl(X_{it}^\perp (X_{it}^\perp)'\bigr)\Bigr)^{-1}
    \mathbb{E}\bigl(X_{it}^\perp Y_{it}\bigr)
   \nonumber\\
    &=
    \Bigl(\mathbb{E}\bigl({\rm Var}(X_{it} \,|\, U_i, V_t)\bigr)\Bigr)^{-1}
    \mathbb{E}\bigl({\rm Var}(X_{it} \,|\, U_i, V_t)\,\beta_{it}\bigr),
    \label{eq::DefBetaStarTwoway_vec}
\end{align}
where ${\rm Var}(X_{it} \,|\, U_i, V_t) = \mathbb{E}\bigl(X_{it}^\perp (X_{it}^\perp)' \,\big|\, U_i, V_t\bigr)$, and the
second equality follows from the law of iterated expectations, 
exactly as in the scalar case. The PC and IFE estimators extend to
\begin{align}
    \widehat{\beta}_{\rm PC} 
    &= \left(\,\sum_{i=1}^n\sum_{t=1}^T 
        \widehat{X_{it}^\perp}\,(\widehat{X_{it}^\perp})'\right)^{-1}
       \sum_{i=1}^n\sum_{t=1}^T 
        \widehat{X_{it}^\perp}\, Y_{it},
    \\[6pt]
    \widehat{\beta}_{\rm IFE} 
    &=\argmin_{\beta \in \mathbb{R}^K} \, \min_{(\lambda, f) \in \mathbb{R}^{nR + TR}} \sum_{i=1}^n \sum_{t=1}^T \left(Y_{it} - X_{it}' \, \beta - \sum_{r=1}^R \lambda_{ir} f_{tr}\right)^2,
\end{align}
where the residuals 
$\widehat{X_{it}^\perp} \in \mathbb{R}^K$ are obtained by applying the  truncated SVD \eqref{eq::PC_firststep} to each of the $K$ 
regressor matrices $X^{(k)} \in \mathbb{R}^{n \times T}$ separately. With those straightforward modifications, 
Theorems~\ref{thm::PC_consistencyNEW} 
and~\ref{thm::IFE_consistencyNEW} continue to hold, that is, one can still show\footnote{
The high-level conditions of Theorems~\ref{thm::PC_consistencyNEW} 
and~\ref{thm::IFE_consistencyNEW} extend to this vector setting   with minimal notational replacements, e.g.\ 
 the non-degeneracy condition~(ii) becomes 
$\frac{1}{nT}\sum_{i,t} X_{it}^\perp (X_{it}^\perp)' \to_P C$ for a positive definite 
matrix $C$.}
\begin{align}
   \widehat \beta_{\rm PC} &=  \beta^* + o_P(1),
&
    \widehat{\beta}_{\rm IFE} &=  \beta^* + o_P(1),
\end{align}
as $n,T \rightarrow \infty$.

The multivariate estimand \eqref{eq::DefBetaStarTwoway_vec} merits careful interpretation.
Consider first $\beta_{it}$ in \eqref{eq::beta_it_twoway_vec}. This is simply the vector of 
conditional linear projection coefficients of $Y_{it}$ on $X_{it}$ given $(U_i, V_t)$: the 
$k$-th component $\beta_{it,k}$ measures the partial effect of $X_{it,k}$ on $Y_{it}$ holding 
the remaining $K-1$ treatments fixed, conditional on the unobserved heterogeneity $(U_i, V_t)$. 
This is a meaningful object: it generalises the scalar conditional projection coefficient 
\eqref{eq::beta_it_twoway} in the natural way, and for binary treatments reduces to the 
conditional average treatment effect of the $k$-th treatment holding the other treatments 
fixed. Importantly, the conditioning on $(U_i, V_t)$ is fully nonparametric here, which is 
more demanding than the partially linear framework of \citet{10.1257/aer.20221116}, who 
condition on observed controls linearly.

The interpretation of $\beta^*$ in \eqref{eq::DefBetaStarTwoway_vec} is more subtle.
Writing out the $k$-th component,
\begin{align}
    \beta^*_k = \sum_{p=1}^K 
    \left[\Bigl(\mathbb{E}\bigl({\rm Var}(X_{it} \,|\, U_i, V_t)\bigr)\Bigr)^{-1}\right]_{kp}
    \Bigl[\mathbb{E}\bigl({\rm Var}(X_{it} \,|\, U_i, V_t)\,\beta_{it}\bigr)\Bigr]_{p},
    \label{eq::betastar_k_contamination}
\end{align}
it is clear that $\beta^*_k$ is generally a function of $\beta_{it,j}$ for all $j$, not only 
$j = k$. Specifically, $\beta^*_k$ can be expressed as
\begin{align}
    \beta^*_k = \mathbb{E}\bigl(\lambda_{kk,it}\,\beta_{it,k}\bigr) 
    + \sum_{j \neq k} \mathbb{E}\bigl(\lambda_{kj,it}\,\beta_{it,j}\bigr),
    \label{eq::contamination_decomp}
\end{align}
where the weights $\lambda_{kj,it}$ are determined by the conditional variance matrix 
${\rm Var}(X_{it} \,|\, U_i, V_t)$ and its expectation. The own-treatment weights satisfy 
$\mathbb{E}(\lambda_{kk,it}) = 1$, while the contamination weights satisfy 
$\mathbb{E}(\lambda_{kj,it}) = 0$ for $j \neq k$.
Since the contamination weights average to zero, they must be negative for some $(i,t)$ 
unless they are identically zero. Hence, $\beta^*_k$ is not a convex combination of the 
$\beta_{it,k}$: it is contaminated by the effects of the other $K-1$ treatments, with weights 
that need not be non-negative. This is precisely the contamination bias phenomenon 
identified by \citet{10.1257/aer.20221116} in cross-sectional regressions with multiple 
treatments and flexible controls, and by \citet{de2023two} in two-way fixed 
effects regressions with several treatments. Our results show that the same phenomenon 
arises in the factor model panel setting, and that it is not an artifact of any particular 
estimation approach: it is a property of the estimand $\beta^*$ itself.

The contamination bias disappears in two special cases that parallel those identified in 
\citet{10.1257/aer.20221116}. First, if the $K$ treatments are conditionally uncorrelated 
given $(U_i, V_t)$, so that ${\rm Var}(X_{it} \,|\, U_i, V_t)$ is diagonal for all $(i,t)$, 
then $\mathbb{E}[{\rm Var}(X_{it} \,|\, U_i, V_t)]$ is also diagonal and the contamination 
weights $\lambda_{kj,it}$ are identically zero. Second, contamination bias is zero if the 
effects $\beta_{it,j}$ are homogeneous across $(i,t)$ for all $j \neq k$, since the 
contamination weights average to zero. Outside these special cases, the vector $\beta^*$ 
retains a well-defined probability limit that is consistently estimated by the PC and IFE 
estimators, but its components do not admit the clean variance-weighted average 
interpretation that makes the scalar estimand \eqref{eq::DefBetaStarTwoway} attractive.

\section{Monte Carlo Analysis}
\label{section:montecarlo}

\subsection{Illustrative DGP}

In this section we consider a simple setting with a single outcome variable $Y_{it}$ 
and a covariate $X_{it}$. Motivated by the framework developed in the previous sections, 
we construct the DGP in the form of a panel location-scale nonlinear factor 
model:\footnote{Alternatively, it can be seen as an extension of the linear panel 
quantile model in \citet{https://doi.org/10.3982/ECTA15746}.}
\begin{align}
Y_{it} &= \gamma_{it} X_{it} + l_y(\lambda_i, f_t) + s_y(\nu_i, g_t)\eps_{it}, \notag\\
X_{it} &= l_x(\lambda_{i,x}, f_{t,x}) + s_x(\nu_{i,x}, g_{t,x})\eps_{it,x}, \notag\\
\eps_{it} &= \rho\,\eps_{it,x} + \sqrt{1-\rho^2}\,u_{it},
\end{align}
where $\eps_{it,x}$ and $u_{it}$ are mutually independent random variables for all 
$(i,t)$. For the location functions $l_y(\lambda_i, f_t)$ and $l_x(\lambda_{i,x}, 
f_{t,x})$, we consider two setups.
\begin{itemize}
    \item \emph{Linear Factor Model} (LFM):
    \begin{equation}
        l_y(\lambda_i, f_t) = \lambda_i f_t, \qquad 
        l_x(\lambda_{i,x}, f_{t,x}) = \lambda_{i,x} f_{t,x}.
    \end{equation}
    \item \emph{Nonlinear Factor Model} (NLFM):
    \begin{equation}
        l_y(\lambda_i, f_t) = (0.5 \times \lambda_i^{10} + 0.5 \times f_t^{10})^{1/10}, 
        \qquad
        l_x(\lambda_{i,x}, f_{t,x}) = (0.5 \times \lambda_{i,x}^{10} + 
        0.5 \times f_{t,x}^{10})^{1/10}.
    \end{equation}
\end{itemize}
The NLFM setup is inspired by the constant elasticity of substitution specification 
for unobserved heterogeneity proposed in \citet{https://doi.org/10.3982/ECTA15238}, 
see also \citet{beyhum2024inferencediscretizingunobservedheterogeneity}, while the LFM 
follows the standard framework of \citet{pesaran2006estimation} and \citet{bai2009panel}.

We restrict our attention to the setting where $\gamma_{it}$ takes the form
\begin{equation}
    \gamma_{it} = \beta_0 + \kappa 
    \frac{s_y(\nu_i, g_t)\,s_x(\nu_{i,x}, g_{t,x})}{(s_x(\nu_{i,x}, g_{t,x}))^2}.
\end{equation}
Note that $\gamma_{it}$ has a causal interpretation as the structural effect of 
$X_{it}$ on $Y_{it}$. We make this choice such that the implied value of $\beta_{it}$ 
for the overall DGP is given by
\begin{equation}
    \beta_{it} = \beta_0 + (\kappa + \rho) 
    \frac{s_y(\nu_i, g_t)\,s_x(\nu_{i,x}, g_{t,x})}{(s_x(\nu_{i,x}, g_{t,x}))^2}.
\end{equation}
This allows us to consider setups where $\beta_{it}$ has either a causal interpretation 
(when $\rho = 0$) or a non-causal one (when $\rho \neq 0$), while keeping the overall 
structure of the estimand constant. In particular, the parameter $\beta^*$ estimated 
by the IFE and PC methods is given by
\begin{equation}
    \beta^* = \beta_0 + (\kappa + \rho) 
    \frac{\mathbb{E}(s_y(\nu_i, g_t)\,s_x(\nu_{i,x}, g_{t,x}))}
         {\mathbb{E}((s_x(\nu_{i,x}, g_{t,x}))^2)},
\end{equation}
which depends on $\kappa$ and $\rho$ only through their sum $\kappa + \rho$. Throughout 
this section we use the linear two-way specification for the scale functions
\begin{equation}
    s_y(\nu_i, g_t) = \nu_i + g_t, \qquad s_x(\nu_{i,x}, g_{t,x}) = \nu_{i,x} + g_{t,x},
\end{equation}
such that the resulting $\beta_{it}$ violates the settings of \citet{chernozhukov2019inference} and \citet{LU2023694}.

The setup involves four unit-specific random variables $U_i = (\lambda_i, \lambda_{i,x}, 
\nu_i, \nu_{i,x})'$ and four time-specific random variables $V_t = (f_t, f_{t,x}, g_t, 
g_{t,x})'$. To keep the results tractable we impose
\begin{equation}
    f_t = g_t, \quad f_{t,x} = g_{t,x}, \quad \nu_i = \lambda_i, \quad 
    \nu_{i,x} = \lambda_{i,x},
\end{equation}
and
\begin{equation}
    \lambda_i = \pi \lambda_i^+ + (1-\pi)\lambda_{i,x}, \qquad 
    f_t = \pi f_t^+ + (1-\pi) f_{t,x}.
\end{equation}
Here $\lambda_i^+$ and $\lambda_{i,x}$ are mutually independent cross-sectionally 
i.i.d.\ $\Gamma(1,1)$ sequences, while $f_t^+$ and $f_{t,x}$ are mutually independent 
AR(1) sequences with innovations drawn from a Gamma distribution with shape parameter 
$(1-\alpha)^2/(1-\alpha^2)$ and scale parameter $(1-\alpha^2)/(1-\alpha)$.\footnote{With 
this parametrization $\mathbb{E}(f_t^+) = \mathbb{E}(f_{t,x}) = 1$ and 
$\mathrm{Var}(f_t^+) = \mathrm{Var}(f_{t,x}) = 1$.
}
We set $\beta_0 = 0$ and 
$\alpha = 0.5$, and vary $\pi \in \{0.0, 0.5\}$.\footnote{Preliminary Monte Carlo 
results indicate that the choice of $\alpha$ is mostly irrelevant as long as 
$|\alpha| < 1$. The value $\beta_0 = 0$ is without loss of generality as long as 
non-zero values of $\rho$ and/or $\kappa$ are included in the analysis.} 
The specification for random variables follows the framework of 
\citet{beyhum2024inferencediscretizingunobservedheterogeneity}. Given the properties 
of the Gamma distribution it is straightforward to establish that
\begin{equation}
    \beta^* = \beta_0 + (\kappa + \rho)\frac{6 - 4\pi}{6}.
\end{equation}

We consider the following four special cases.
\begin{description}
    \item[(DGP.1)] NC-LFM. $\kappa = 0$ and $\rho = 0.5$. LFM.
    \item[(DGP.2)] NC-NLFM. $\kappa = 0$ and $\rho = 0.5$. NLFM.
    \item[(DGP.3)] C-LFM. $\kappa = 0.5$ and $\rho = 0.0$. LFM.
    \item[(DGP.4)] C-NLFM. $\kappa = 0.5$ and $\rho = 0.0$. NLFM.
\end{description}
In the first two designs $\beta_{it}$ has no causal (NC) interpretation, while in 
the last two it does, hence the label C.\footnote{The CCE estimator of 
\citet{pesaran2006estimation} is consistent for $\beta^*$ and asymptotically normal 
only under (DGP.1).}

Two remarks are in order. First, $\beta^*$ is a function of $l_y(\lambda_i, f_t)$ 
and $l_x(\lambda_{i,x}, f_{t,x})$, but not of $s_y(\nu_i, g_t)$ and 
$s_x(\nu_{i,x}, g_{t,x})$. As a result, $E_{it}^\parallel = Y_{it}^\parallel - \beta^* 
X_{it}^\parallel$ is not affected by the dimensionality of the unobservables in the 
scaling functions. Second, the dimensions of unobservables in $E_{it}^\parallel$ are 
larger than those in $X_{it}^\parallel$ (two-dimensional for both the cross-sectional 
and time series dimensions in $E_{it}^\parallel$, versus one-dimensional for 
$X_{it}^\parallel$). In this regard the PC approach is advantageous relative to the 
IFE method.

We set $n = T$ and vary $n \in \{25, 50, 75, 100, 200\}$. The number of factors 
$R_{nT}$ is set following \citet{FreemanWeidner2021} to $R_{nT} = \lfloor 3n^{3/8} 
\rfloor$. The number of Monte Carlo replications is $M = 10000$, and all results are 
reported for the normalized estimator $\sqrt{\min(n,T)}(\widehat{\beta} - \beta^*)$.

\subsection{Results}

      \begin{table}[htbp!]
  \renewcommand{\arraystretch}{.9}
  \addtolength{\tabcolsep}{-2pt}
  \centering
 \footnotesize
  \caption{Monte Carlo results. DGP.1 NC-LFM $\kappa=0.0$ and $\rho=0.5$. LFM.}
\begin{tabular}{cc|cc|cc|cc|cc}
\hline
\hline
&&\multicolumn{2}{c|}{IFE}&\multicolumn{2}{c|}{PC (YX)}&\multicolumn{2}{c|}{PC (X)}&\multicolumn{2}{c}{CCE (YX)}\\
     $n=T$  &    $\pi$  &    Bias  &    Var  &     Bias  &    Var  &     Bias  &    Var  &     Bias  &    Var \\
\hline
\hline

25 & 0.0 & 0.010 & 0.050 & -0.703 & 0.013 & 0.005 & 0.053 & 0.005 & 0.018 \\
  25 & 0.5 & 0.219 & 0.054 & -0.336 & 0.018 & 0.658 & 0.137 & 0.236 & 0.027 \\
 50 & 0.0 & 0.010 & 0.011 & -0.622 & 0.004 & 0.001 & 0.010 & 0.003 & 0.007 \\
 50 & 0.5 & 0.229 & 0.014 & -0.203 & 0.008 & 0.596 & 0.040 & 0.253 & 0.014 \\
 75 & 0.0 & 0.009 & 0.005 & -0.574 & 0.002 & 0.001 & 0.004 & 0.002 & 0.004 \\
 75 & 0.50 & 0.248 & 0.008 & -0.129 & 0.006 & 0.577 & 0.024 & 0.269 & 0.010 \\
100 & 0.0 & 0.007 & 0.003 & -0.529 & 0.002 & 0.002 & 0.003 & 0.002 & 0.002 \\
 100 & 0.5 & 0.264 & 0.006 & -0.067 & 0.005 & 0.562 & 0.016 & 0.283 & 0.008 \\
 200 & 0.0 & 0.003 & 0.001 & -0.470 & 0.001 & 0.000 & 0.001 & 0.001 & 0.001 \\
 200 & 0.5 & 0.317 & 0.004 &  0.040 & 0.004 & 0.568 & 0.009 & 0.325 & 0.005 \\
\hline
\hline
\end{tabular}
    \label{tab::DGP1}
\end{table}

      \begin{table}[htbp!]
  \renewcommand{\arraystretch}{.9}
  \addtolength{\tabcolsep}{-2pt}
  \centering
 \footnotesize
  \caption{Monte Carlo results. DGP.2 NC-NLFM $\kappa=0.0$ and $\rho=0.5$. NLFM.}
\begin{tabular}{cc|cc|cc|cc|cc}
\hline
\hline
&&\multicolumn{2}{c|}{IFE}&\multicolumn{2}{c|}{PC (YX)}&\multicolumn{2}{c|}{PC (X)}&\multicolumn{2}{c}{CCE (YX)}\\
     $n=T$  &    $\pi$  &    Bias  &    Var  &     Bias  &    Var  &     Bias  &    Var  &     Bias  &    Var \\
\hline
\hline

25 & 0.0 & 0.059 & 0.051 & -0.700 & 0.013 & 0.025 & 0.052 & 0.043 & 0.018 \\
 25 & 0.5 & 0.227 & 0.051 & -0.347 & 0.017 & 0.629 & 0.126 & 0.230 & 0.026 \\
 50 & 0.0 & 0.049 & 0.011 & -0.619 & 0.004 & 0.014 & 0.010 & 0.035 & 0.007 \\
50 & 0.5 & 0.241 & 0.013 & -0.212 & 0.008 & 0.578 & 0.038 & 0.248 & 0.014 \\
 75 & 0.0 & 0.042 & 0.005 & -0.570 & 0.002 & 0.012 & 0.005 & 0.029 & 0.004 \\
 75 & 0.5 & 0.258 & 0.008 & -0.136 & 0.006 & 0.564 & 0.023 & 0.265 & 0.010 \\
100 & 0.0 & 0.036 & 0.003 & -0.525 & 0.002 & 0.010 & 0.003 & 0.025 & 0.003 \\
100 & 0.5 & 0.273 & 0.006 & -0.074 & 0.005 & 0.551 & 0.016 & 0.279 & 0.008 \\
200 & 0.0 & 0.022 & 0.001 & -0.465 & 0.001 & 0.007 & 0.001 & 0.018 & 0.001 \\
 200 & 0.5 & 0.318 & 0.004 &  0.035 & 0.004 & 0.560 & 0.008 & 0.322 & 0.005 \\

\hline
\hline
\end{tabular}
    \label{tab::DGP2}
\end{table}

      \begin{table}[htbp!]
  \renewcommand{\arraystretch}{.9}
  \addtolength{\tabcolsep}{-2pt}
  \centering
 \footnotesize
  \caption{Monte Carlo results. DGP.3 C-LFM $\kappa=0.5$ and $\rho=0.0$. LFM.}
\begin{tabular}{cc|cc|cc|cc|cc}
\hline
\hline
&&\multicolumn{2}{c|}{IFE}&\multicolumn{2}{c|}{PC (YX)}&\multicolumn{2}{c|}{PC (X)}&\multicolumn{2}{c}{CCE (YX)}\\
     $n=T$  &    $\pi$  &    Bias  &    Var  &     Bias  &    Var  &     Bias  &    Var  &     Bias  &    Var \\
\hline
\hline

 25 & 0.0 & 0.007 & 0.065 & -0.721 & 0.016 & 0.008 & 0.071 & 0.007 & 0.024 \\
 25 & 0.5 & 0.237 & 0.067 & -0.350 & 0.021 & 0.663 & 0.160 & 0.238 & 0.031 \\
 50 & 0.0 & 0.005 & 0.014 & -0.644 & 0.005 & 0.003 & 0.013 & 0.004 & 0.009 \\
 50 & 0.5 & 0.243 & 0.017 & -0.217 & 0.009 & 0.600 & 0.043 & 0.254 & 0.015 \\
 75 & 0.0 & 0.003 & 0.006 & -0.596 & 0.003 & 0.003 & 0.006 & 0.002 & 0.005 \\
 75 & 0.5 & 0.258 & 0.009 & -0.142 & 0.006 & 0.580 & 0.025 & 0.269 & 0.011 \\
100 & 0.0 & 0.002 & 0.003 & -0.550 & 0.002 & 0.003 & 0.004 & 0.002 & 0.003 \\
100 & 0.5 & 0.271 & 0.007 & -0.079 & 0.005 & 0.564 & 0.017 & 0.283 & 0.009 \\
 200 & 0.0 & 0.001 & 0.001 & -0.490 & 0.001 & 0.001 & 0.001 & 0.001 & 0.001 \\
200 & 0.5 & 0.314 & 0.004 &  0.029 & 0.004 & 0.570 & 0.009 & 0.325 & 0.005 \\

\hline
\hline
\end{tabular}
    \label{tab::DGP3}
\end{table}

      \begin{table}[htbp!]
  \renewcommand{\arraystretch}{.9}
  \addtolength{\tabcolsep}{-2pt}
  \centering
 \footnotesize
  \caption{Monte Carlo results. DGP.4 C-NLFM $\kappa=0.5$ and $\rho=0.0$. NLFM.}
\begin{tabular}{cc|cc|cc|cc|cc}
\hline
\hline
&&\multicolumn{2}{c|}{IFE}&\multicolumn{2}{c|}{PC (YX)}&\multicolumn{2}{c|}{PC (X)}&\multicolumn{2}{c}{CCE (YX)}\\
     $n=T$  &    $\pi$  &    Bias  &    Var  &     Bias  &    Var  &     Bias  &    Var  &     Bias  &    Var \\
\hline
\hline

 25 & 0.0 & 0.055 & 0.070 & -0.715 & 0.016 & 0.050 & 0.071 & 0.069 & 0.026 \\
 25 & 0.5 & 0.248 & 0.065 & -0.358 & 0.020 & 0.647 & 0.150 & 0.245 & 0.030 \\
 50 & 0.0 & 0.034 & 0.014 & -0.637 & 0.005 & 0.030 & 0.014 & 0.052 & 0.009 \\
 50 & 0.5 & 0.257 & 0.016 & -0.223 & 0.009 & 0.591 & 0.042 & 0.259 & 0.015 \\
 75 & 0.0 & 0.025 & 0.006 & -0.589 & 0.003 & 0.023 & 0.006 & 0.043 & 0.005 \\
 75 & 0.5 & 0.273 & 0.009 & -0.146 & 0.006 & 0.574 & 0.025 & 0.274 & 0.011 \\
100 & 0.0 & 0.020 & 0.004 & -0.543 & 0.002 & 0.020 & 0.004 & 0.037 & 0.003 \\
100 & 0.5 & 0.286 & 0.007 & -0.082 & 0.005 & 0.560 & 0.017 & 0.288 & 0.009 \\
 200 & 0.0 & 0.013 & 0.001 & -0.482 & 0.001 & 0.013 & 0.001 & 0.028 & 0.001 \\
200 & 0.5 & 0.331 & 0.004 &  0.028 & 0.004 & 0.567 & 0.009 & 0.330 & 0.005 \\

\hline
\hline
\end{tabular}
    \label{tab::DGP4}
\end{table}

The results in Tables~\ref{tab::DGP1}--\ref{tab::DGP4} are broadly similar across all 
four DGPs, for all estimators considered. This is perhaps surprising for the CCE 
estimator, which is consistent for $\beta^*$ only under (DGP.1). The bias appears to 
be of order $\largeO(1)$ while the variance diminishes as $n, T \to \infty$, consistent 
with the theoretical finding that the asymptotic distribution is dominated by 
approximation bias. The value of $\pi$, which determines $\beta^*$, has a large 
effect on the bias. The IFE, CCE, and PC(X) estimators are positively biased for 
larger values of $\pi$, while the PC(YX) estimator has negative bias when $\pi = 0$. 
An intuitive explanation for the behaviour of PC(X) and PC(YX) can be related to the 
bias expression in \citet{westerlund2015cross} and the non-invariance of PC estimators 
to $\beta^*$ and the correlations between the factor loadings and the errors. Finally, 
the finite-sample distributions of the estimators appear close to normal, as 
illustrated in Figure~\ref{fig:MC1}.

\begin{figure}[htpb]
\centering
\begin{subfigure}{.5\textwidth}
  \centering
  \includegraphics[scale=0.55]{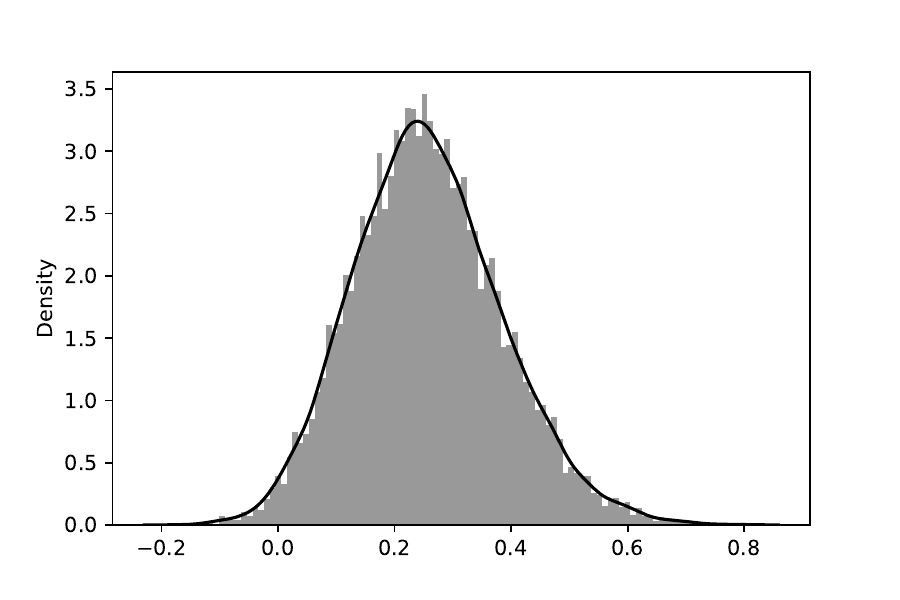}
  \caption{IFE.}
\end{subfigure}%
\begin{subfigure}{.5\textwidth}
  \centering
  \includegraphics[scale=0.55]{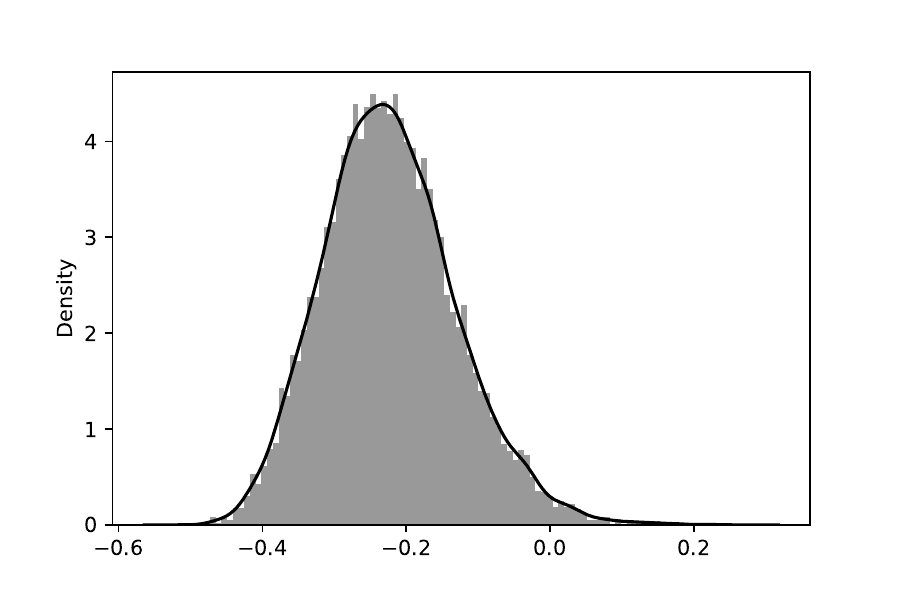}
  \caption{PC(YX).}
\end{subfigure}
\caption{Finite-sample distributions for $n = T = 50$ under (DGP.4) with $\pi = 0.5$.}
\label{fig:MC1}
\end{figure}
 
\section{Conclusions}
\label{section:conclusions}

This paper shows that two widely used large panel estimators, the PC estimator of 
\citet{GreenawayMcGrevy201248} and the IFE estimator of \citet{bai2009panel}, have 
well-defined and interpretable probability limits under fully nonparametric assumptions 
on the data generating process. Specifically, both estimators converge to the same 
variance-weighted average of unit-time-specific treatment effects $\beta_{it}$, where 
the weights are proportional to the conditional variance of the regressor given the 
unobserved heterogeneity. This result requires no parametric structure on the 
relationship between the outcome, the regressor, and the unobserved heterogeneity. 
The key requirement is that the number of estimated factors $R_{nT}$ grows with the 
sample size, but slowly enough that estimation error remains controlled. This is a 
meaningful departure from the existing literature, which typically requires $R_{nT}$ 
to be fixed and the factor structure to be correctly specified. Thus, even if the 
underlying linear panel regression model with interactive fixed effects is misspecified, 
one can still attach a clear causal or descriptive meaning to the resulting estimates.

The main limitation of our results is that they apply to the case of a single 
regressor. When multiple regressors are included, the variance-weighted average 
interpretation breaks down due to contamination bias, as we discuss in 
Section~\ref{ssection::additional_regressors}. A second limitation concerns inference. 
Our consistency results are silent on the distribution of the estimators, and we 
argue in Section~\ref{ssection:inference} that the asymptotic distribution is 
dominated by approximation bias, making inference on $\beta^*$ generally impossible 
without additional parametric assumptions.

\section*{Acknowledgment}
This paper benefited from the use of generative AI tools to assist with 
proofreading, language editing, and \LaTeX{} formatting. All output was carefully reviewed by the authors. All substantive 
content, results, and any remaining errors are the authors' responsibility.

\subsection*{Competing interests}
The authors declare none.

\setlength{\bibsep}{2pt}


\clearpage

\appendix
\section*{Appendix}
\section{Additional Material}
\subsection{TWFE and CCE do not estimate relevant estimands}
\label{ssection::appendix_TWFE_CCE}
In what follows we briefly explain why, unlike the PC based estimators analyzed above, 
two other popular approaches do not generally estimate the variance-weighted average of 
$\beta_{it}$. The CCE estimator of \citet{pesaran2006estimation} generally fails because 
the number of cross-section averages is fixed by assumption, and they only capture the 
information in the conditional mean of $(Y_{it}, X_{it})$ up to that fixed rank. The 
TWFE estimator fails for a different reason: by construction it can only control for 
additive conditional expectation functions.

A counterexample which proves the above claim is given by the simple one-factor model with time effects
\begin{align}
     Y_{it} &=  \eta_{i}+\lambda_i \, z_t + \varepsilon_{it} ,\\
     X_{it} &=  \kappa_{i}+\lambda_i \, z_t + \nu_{it} ,
 \end{align}
where $\lambda_i$, $z_t$, $\eta_i$, $\kappa_i$, $\varepsilon_{it}$, $\nu_{it}$ are 
mutually independent, each i.i.d.\ over their respective indices, with mean zero 
and unit variance. In this example, $\beta_{it}=0$ a.s. as $\varepsilon_{it}$ and $\nu_{it}$ are independent. Thus, also $\beta_{w}=0$ for any choice of weights $(w_{it}: i\geq1,t\geq 1)$.
 
 Now consider the CCE and TWFE estimands in this example. As $\mathbb{E}(\lambda_{i})=0$, the factor $z_{t}$ cannot be controlled for by the cross-section averages underlying the CCE approach. Obviously, the same holds for the TWFE estimator as it completely neglects the factor component in the estimator. Moreover, for this specific example the two estimators have the same probability limit, i.e.
 \begin{align*}
\beta_{\rm TWFE}^{*}=\beta_{\rm CCE}^{*}=\frac{\mathbb{E}(\lambda_{i}^{2})\mathbb{E}(z_{t}^{2})}{\mathbb{E}(\lambda_{i}^{2})\mathbb{E}(z_{t}^{2})+\mathbb{E}(\nu_{it}^{2})}=\frac{1}{2}.
 \end{align*}
These results can be obtained directly from the corresponding derivations in \citet{JuodisRCCE2020,Juodis2020}.
\subsection{Further results conditional on $U$ and $V$}
\label{appendix::conditional_U_V}
The following lemma provides a precise sufficient condition for Condition (iii) of Theorem~\ref{thm::conditional}.
\begin{lemma}[\bf Sufficient conditions for low-rank approximation]
\label{lemma::lowrank_conditional}
Let $\mathcal{U} \subset \mathbb{R}^{d_U}$ and $\mathcal{V} \subset \mathbb{R}^{d_V}$ be compact sets, and let $h : \mathcal{U} \times \mathcal{V} \to \mathbb{R}$ be $s$-times continuously differentiable with $s > \max(d_U, d_V)$. Let $h(u,v) = \sum_{r=1}^\infty \sigma_r \phi_r(u) \psi_r(v)$ be the singular value decomposition of $h$ with respect to the uniform measure on $\mathcal{U} \times \mathcal{V}$, where $\sigma_1 \geq \sigma_2 \geq \ldots \geq 0$, and $\{\phi_r\}$ and $\{\psi_r\}$ are orthonormal. Suppose the sequences $(U_1, \ldots, U_n) \in \mathcal{U}^n$ and $(V_1, \ldots, V_T) \in \mathcal{V}^T$ satisfy
\begin{align}
    \sup_{r \geq 1} \frac{1}{n} \sum_{i=1}^n \phi_r(U_i)^2 \leq C, \qquad \sup_{r \geq 1} \frac{1}{T} \sum_{t=1}^T \psi_r(V_t)^2 \leq C,
    \label{eq::bounded_singular_functions}
\end{align}
for some constant $C > 0$. Define the $n \times T$ matrix $H$ by $H_{it} = h(U_i, V_t)$. Then
\[
\min_{{\rm rank}(G) \leq R} \frac{1}{nT} \|H - G\|_2^2 \leq C^2 \left(\sum_{r > R} \sigma_r\right)^2 = o(1)
\]
as $R \to \infty$.
\end{lemma}

Condition \eqref{eq::bounded_singular_functions} requires that the empirical second moments of the singular functions remain bounded uniformly in $r$. This condition is automatically satisfied when $(U_i)$ are i.i.d.\ draws from the uniform distribution on $\mathcal{U}$, since then $\frac{1}{n}\sum_i \phi_r(U_i)^2 \to_P \int \phi_r(u)^2 \, du = 1$ by orthonormality. More generally, the condition holds for sequences $(U_i)$ that are ``well-spread'' in $\mathcal{U}$, such as points on a regular grid or low-discrepancy sequences. It rules out pathological cases where the $(U_i)$ concentrate on a low-dimensional subspace where some singular functions happen to be large.

\subsection{Convergence rate of IFE estimator}
\label{ssection::freeman_weidner}
Below we provide a theorem that sheds more light on the convergence rate of the IFE estimator.
\begin{theorem}[\bf Convergence rate of IFE estimator]
\label{thm::IFE_convergencerateNEW}
 Let  $\beta^{*} \in \mathbb{R}$ be a constant
      and  $E^\parallel \in \mathbb{R}^{n \times T}$ 
      be a random matrix, and define  $E:=Y - \beta^* \, X$
      and $E^\perp:=E- E^\parallel$.
     Assume that, as $n,T \rightarrow \infty$ we have $R_{nT}=o(\min\{n,T\})$ and

    \begin{enumerate}[(i)]

     \item    
     $\frac 1 {nT} \sum_{i=1}^n \sum_{t=1}^T  X_{it} \, E^\perp_{it}  ={\cal O}_P((\min(n,T))^{-1/2})$,
     and     
     $\frac 1 {nT} \sum_{i=1}^n \sum_{t=1}^T    X_{it} ^2  ={\cal O}_{P}(1)$.
       
       \item 
        $ 
     \min_{\big\{ G \in \mathbb{R}^{n \times T}: {\rm rank}(G) \leq 2R_{nT} \big\}} 
     \frac{1}{nT}\left\|  X -  G \right\|^2_2 
  \geq c$, wpa1, for a constant $c>0$.

       \item $\| E^{\perp} \|={\cal O}_{P}\left(\sqrt{n}+\sqrt{T\log^{2} T}\right)$.
    
       \item There exists a constant $\rho> 3/2$ such that
\begin{equation*}
\min_{\big\{ G \in \mathbb{R}^{n \times T}: {\rm rank}(G) \leq R_{nT} \big\}} 
     \frac{1}{nT}\left\|  E^{\parallel} -  G \right\|^2_2={\cal O}_{P}\left(R_{nT}^{1-2\rho}\right).
\end{equation*}

    \end{enumerate}
    Then we have:
\begin{equation}
\widehat{\beta}_{\rm IFE}-\beta^{*}={\cal O}_P((\min(n,T))^{-1/2})+{\cal O}_{P}(R_{nT}^{(3-2\rho)/2})+{\cal O}_{P}(R_{nT} (\min(T,n/\log^{2}T))^{-1/2}).
\end{equation}
Therefore, by choosing $R_{nT}\propto (\min(T,n/\log^{2}T))^{\frac{1}{2\rho-1}}$ we obtain that:
\begin{equation}
\label{eq::rate_with_optimal}
\widehat{\beta}_{\rm IFE}-\beta^{*}={\cal O}_P((\min(n,T))^{-1/2})+{\cal O}_{P}((\min(T,n/\log^{2}T))^{\frac{3-2\rho}{2(2\rho-1)}}).
\end{equation}

\end{theorem}

The proof of this theorem is based on the proof strategy used in Lemma 1 and Theorem 1 of \citet{FreemanWeidner2021}. At the same time, the rates imposed above differ from the corresponding rates in \citet{FreemanWeidner2021} due to the non-causal nature of the model of the decomposition $E:=Y - \beta^* \, X$. In particular, the term $\largeO_P((\min(n,T))^{-1/2})$ is present in this theorem, but not present in \citet{FreemanWeidner2021}. This term is a direct consequence of $(i)$ that cannot be generally improved upon for misspecified panel data models, see e.g. \citet{Juodis2020}. Secondly, for the reasons discussed in the main body of the paper, there is a $\log(T)$ term present in $(iii)$, and, as a result, also in the main statement of the theorem. If we were to adopt the same assumptions as \citet{FreemanWeidner2021}, the convergence rate would be the same.

For example, if $n=T$ (and ignoring the $\log(T)$ term), it is evident that the second term in \eqref {eq::rate_with_optimal} (associated with the truncation error) will never be dominated.

\section{Further Monte Carlo Results}
      \begin{table}[htbp!]
  \renewcommand{\arraystretch}{.9}
  \addtolength{\tabcolsep}{-2pt}
  \centering
 \footnotesize
  \caption{Monte Carlo results. DGP.5 NC-NLFM $\kappa=0.0$ and $\rho=0.5$. NLFM as in \citet{FreemanWeidner2021}.}
\begin{tabular}{cc|cc|cc|cc|cc}
\hline
\hline

&&\multicolumn{2}{c|}{IFE}&\multicolumn{2}{c|}{PC (YX)}&\multicolumn{2}{c|}{PC (X)}&\multicolumn{2}{c}{CCE (YX)}\\
     $n=T$  &    $\pi$  &    Bias  &    Var  &     Bias  &    Var  &     Bias  &    Var  &     Bias  &    Var \\

\hline
\hline
25 & 0.0 & 0.363 & 0.057 & -0.613 & 0.013 & 0.176 & 0.044 & 0.134 & 0.019 \\
 25 & 0.5 & 0.072 & 0.042 & -0.448 & 0.013 & 0.382 & 0.099 & 0.140 & 0.017 \\
  50 & 0.0 & 0.334 & 0.017 & -0.524 & 0.004 & 0.159 & 0.009 & 0.142 & 0.008 \\
  50 & 0.5 & 0.094 & 0.010 & -0.322 & 0.005 & 0.396 & 0.028 & 0.161 & 0.008 \\
 75 & 0.0 & 0.314 & 0.009 & -0.479 & 0.003 & 0.145 & 0.004 & 0.147 & 0.005 \\
 75 & 0.5 & 0.115 & 0.005 & -0.243 & 0.004 & 0.407 & 0.017 & 0.174 & 0.006 \\
 100 & 0.0 & 0.300 & 0.005 & -0.434 & 0.002 & 0.135 & 0.003 & 0.154 & 0.003 \\
 100 & 0.5 & 0.131 & 0.004 & -0.174 & 0.004 & 0.414 & 0.012 & 0.186 & 0.005 \\
 200 & 0.0 & 0.278 & 0.002 & -0.383 & 0.001 & 0.105 & 0.001 & 0.167 & 0.002 \\
200 & 0.5 & 0.171 & 0.002 & -0.047 & 0.003 & 0.453 & 0.007 & 0.214 & 0.003 \\
\hline
\hline

\end{tabular}
    \label{tab::DGP5}
\end{table}
      \begin{table}[htbp!]
  \renewcommand{\arraystretch}{.9}
  \addtolength{\tabcolsep}{-2pt}
  \centering
 \footnotesize
  \caption{Monte Carlo results. DGP.6 C-NLFM $\kappa=0.5$ and $\rho=0.0$. NLFM as in \citet{FreemanWeidner2021}.}
\begin{tabular}{cc|cc|cc|cc|cc}
\hline
\hline

&&\multicolumn{2}{c|}{IFE}&\multicolumn{2}{c|}{PC (YX)}&\multicolumn{2}{c|}{PC (X)}&\multicolumn{2}{c}{CCE (YX)}\\
     $n=T$  &    $\pi$  &    Bias  &    Var  &     Bias  &    Var  &     Bias  &    Var  &     Bias  &    Var \\

\hline
\hline
  25 & 0.0 & 0.486 & 0.124 & -0.593 & 0.018 & 0.353 & 0.065 & 0.262 & 0.035 \\
  25 & 0.5 & 0.322 & 0.069 & -0.341 & 0.020 & 0.880 & 0.257 & 0.368 & 0.046 \\
  50 & 0.0 & 0.402 & 0.026 & -0.503 & 0.006 & 0.317 & 0.014 & 0.276 & 0.016 \\
  50 & 0.5 & 0.319 & 0.018 & -0.201 & 0.009 & 0.815 & 0.077 & 0.395 & 0.024 \\
   75 & 0.0 & 0.350 & 0.011 & -0.460 & 0.003 & 0.289 & 0.007 & 0.285 & 0.010 \\
  75 & 0.5 & 0.329 & 0.010 & -0.123 & 0.006 & 0.773 & 0.041 & 0.413 & 0.018 \\
100 & 0.0 & 0.318 & 0.006 & -0.415 & 0.002 & 0.270 & 0.004 & 0.299 & 0.008 \\
 100 & 0.5 & 0.340 & 0.007 & -0.055 & 0.005 & 0.741 & 0.025 & 0.434 & 0.014 \\
200 & 0.0 & 0.233 & 0.002 & -0.371 & 0.001 & 0.210 & 0.001 & 0.326 & 0.004 \\
 200 & 0.5 & 0.372 & 0.004 &  0.058 & 0.004 & 0.694 & 0.011 & 0.482 & 0.009 \\
\hline
\hline

\end{tabular}
    \label{tab::DGP6}
\end{table}

\clearpage

\section{Proofs}
\label{section_appendix:proofs}

\subsection{Proofs for Section~\ref{ssection::consistency}}

\begin{proof}[\bf Proof of Theorem~\ref{thm::PC_consistencyNEW}]
\# \underline{Step 1: analyze the finite rank approximation to $X^{\parallel}$:}
We write $R$ instead of $R_{nT}$ throughout this proof.
Define
$$ 
  X^\parallel_R := \argmin_{\big\{ G \in \mathbb{R}^{n \times T}: {\rm rank}(G) \leq R \big\}}  \left\|  X^\parallel -  G \right\|^2_2 ,
$$
which by 
assumption (iv) of the theorem satisfies
$$
   \frac 1 {nT}  \left\|  X^\parallel -   X^\parallel_R \right\|^2_2 
   = o_P(1) .
$$
Using this and assumption (ii) of the theorem we find that
\begin{align*}
  \left| {\rm Tr}\left[\left(X^\parallel- X^\parallel_R \right)' X^\perp \right]
  \right| &\leq   \underbrace{ \left\|  X^\parallel -   X^\parallel_R \right\|_2 }_{=o_P(\sqrt{nT})}
     \underbrace{ \left\| X^\perp \right\|_2 }_{=\sqrt{cnT} + o_P(\sqrt{nT})}
     \\
     &= o_P(nT) .
\end{align*}
Since $X^\parallel_R$ is defined as the leading $R$ principal components of $X^\parallel$ we have ${\rm rank}(X^\parallel_R) \leq R$
and $ \left\| X^\parallel_R \right\|_2  \leq  \left\| X^\parallel  \right\|_2 $.\footnote{
If we would have $ \left\| X^\parallel_R \right\|_2  >  \left\| X^\parallel  \right\|_2 $, then it's easy to see that $G = \rho X^\parallel$ with 
$\rho= \frac{{\rm Tr}\left[(X^\parallel_R)'X^\parallel \right]} {\left\| X^\parallel_R \right\|_2^2}$ would give ${\rm rank}(G) \leq R$ and
$  \left\|  X^\parallel -  G \right\|^2_2 <  \left\|  X^\parallel -   X^\parallel_R \right\|^2_2$, which contradicts the definition $X^\parallel_R$ as the $\argmin$ of this objective function.
}
Using this as well as assumptions (i), (ii), (iii) in the theorem, and writing 
$\|\cdot\|_1$ for the nuclear norm (the sum of the singular values), we find that
\begin{align*}
     \left| {\rm Tr}\left[ \left(X^\parallel_R \right)' X^\perp \right] 
    \right|
    &\leq    \left\| X^\perp \right\| \,  \left\|X^\parallel_R \right\|_1 
    \\
    &\leq   \left\| X^\perp \right\| \,
     \left[{\rm rank}(X^\parallel_R)\right]^{1/2}
     \, \left\| X^\parallel_R \right\|_2 
   \\
    &\leq  R^{1/2} \, \left\| X^\perp \right\| \,  \left\| X^\parallel  \right\|_2 
   \\
    &\leq  R^{1/2} \, \left\| X^\perp \right\| \,  \left\| X - X^\perp \right\|_2 
   \\
    &\leq  \underbrace{ R^{1/2} \, \left\| X^\perp \right\| }_{=o_P(\sqrt{nT})} \, \underbrace{ \left( \left\| X \right\|_2 + \left\| X^\perp \right\|_2 \right)   }_{=\largeO_P(\sqrt{nT})}
   \\
   &= o_P(nT) .
\end{align*}
Combining the last two equations gives
\begin{align*}
   \frac 1 {nT} {\rm Tr}\left[ \left(X^\parallel \right)' X^\perp \right]
   &= 
    \frac 1 {nT}
    {\rm Tr}\left[ \left(X^\parallel_R \right)' X^\perp \right]
   +\frac 1 {nT} {\rm Tr}\left[\left(X^\parallel- X^\parallel_R \right)' X^\perp \right]
   \\
   &= o_P(1).
\end{align*}

\bigskip
\noindent
\# \underline{Step 2: show that $\frac 1 {nT}\left\| \widehat{X^\perp} - X^\perp \right\|_2^2 = o_P(1)$:}
Remember that $\widehat{X^\perp} :=X - \widehat X_R$, where
\begin{align*}
      \widehat X_R =  \argmin_{\left\{ G  \in \mathbb{R}^{n \times T} \, : \, {\rm rank}(G) \leq R  \right\}}
       \left\|  X -  G    \right\|_2^2 .
\end{align*}
Since $\widehat X_R$ is defined as the leading $R$ principal components of $X$ we have ${\rm rank}(\widehat X_R) \leq R$
and $ \left\| \widehat X_R \right\|_2  \leq  \left\| X  \right\|_2 $.
 Using this 
as well as assumptions (i) and (iii) in the theorem we find that
\begin{align}
     \left| {\rm Tr}\left[ \left(\widehat X_R \right)' X^\perp \right] 
    \right|
    &\leq   \left\| X^\perp \right\| \,  \left\|\widehat X_R \right\|_1
   \nonumber \\
    &\leq   \left\| X^\perp \right\| \,\left[ {\rm rank}(\widehat X_R) \right]^{1/2}
    \,
    \left\| \widehat X_R \right\|_2 
   \nonumber \\
    &\leq  \underbrace{ R^{1/2} \, \left\| X^\perp \right\| }_{=o_P(\sqrt{nT})} \, \underbrace{   \left\| X \right\|_2    }_{=\largeO_P(\sqrt{nT})}
     \nonumber \\
   &= o_P(nT) .
   \label{ResultsXRXperp}
\end{align}
Since $G=X^\parallel_R$ is a possible choice for $G$
in the above minimization problem that defines $ \widehat X_R$
we find 
\begin{align*}
   \frac 1 {nT} \, \left\|  X -   \widehat X_R    \right\|_2^2
    & \leq   \frac 1 {nT} \,  \left\|  X -   X^\parallel_R   \right\|_2^2
    \nonumber \\
    &=   \frac 1 {nT} \,  \left\|  X^\perp  + \left( X^\parallel -   X^\parallel_R \right)   \right\|_2^2
    \nonumber \\
    &= \frac 1 {nT} \,  \left\|   X^\perp    \right\|_2^2 +  \underbrace{  \frac 1 {nT} \,  \left\|   X^\parallel -   X^\parallel_R   \right\|_2^2 }_{=o_P(1)}
    + \underbrace{ \frac 2 {nT}  {\rm Tr}\left[ \left( X^\parallel -   X^\parallel_R \right)'  X^\perp \right] }_{=o_P(1)} ,
     \nonumber \\
        &= \frac 1 {nT} \,  \left\|   X^\perp    \right\|_2^2 + o_P(1),
\end{align*}
where we used results proved in Step 1 above. 
Using the decomposition
$X =    X^\perp  + X^\parallel  $ we also find that
\begin{align*}
     \frac 1 {nT} \, \left\|  X -   \widehat X_R    \right\|_2^2
     &=   \frac 1 {nT} \, \left\|  X^\perp +     \left(X^\parallel   -   \widehat X_R \right)    \right\|_2^2
    \\
    &= \frac 1 {nT} \,  \left\|   X^\perp    \right\|_2^2
    + \frac 1 {nT} \,  \left\|  X^\parallel   -   \widehat X_R    \right\|_2^2
    + \underbrace{ \frac 2 {nT}  {\rm Tr}\left[  \left(  X^\parallel \right)'  X^\perp \right] }_{=o_P(1)}
    + \underbrace{ \frac 2 {nT}  {\rm Tr}\left[  \left(   \widehat X_R \right)'  X^\perp \right] }_{=o_P(1)} 
   \\
   &= \frac 1 {nT} \,  \left\|   X^\perp    \right\|_2^2
    + \frac 1 {nT} \,  \left\|  X^\parallel   -   \widehat X_R    \right\|_2^2 + o_P(1) ,
\end{align*}
where we used \eqref{ResultsXRXperp} and again results proved in Step 1. Combining the last two results we find that $\frac 1 {nT} \,  \left\|  X^\parallel   -   \widehat X_R    \right\|_2^2
    \leq o_P(1)$, which is of course equivalent to
\begin{align*}
    \frac 1 {nT} \,  \left\|  X^\parallel   -   \widehat X_R    \right\|_2^2
    = o_P(1) .
\end{align*}
Finally, since $X= X^\perp  + X^\parallel  = \widehat{X^\perp} +  \widehat X_R$ we find
$
    X^\parallel   -   \widehat X_R  = \widehat{X^\perp} - X^\perp ,
$
and therefore
\begin{align*}
  \frac 1 {nT}\left\| \widehat{X^\perp} - X^\perp \right\|_2^2 
  &= \frac 1 {nT} \,  \left\|  X^\parallel   -   \widehat X_R    \right\|_2^2
  \\
    &= o_P(1),
\end{align*}
which is what we wanted to show in this step of the proof.

\bigskip

\noindent
\# \underline{Step 3: conclude that $\widehat \beta_{\rm PC} = \beta^* + o_P(1)$:}
Using that  $\frac 1 {nT}\left\| \widehat{X^\perp} - X^\perp \right\|_2^2 = o_P(1)$ we find that
\begin{align*}
  \left| \frac 1 {nT} \left\|  \widehat{X^\perp} \right\|_2^2  -
  \frac 1 {nT} \left\| X^\perp \right\|_2^2 \right|^2
  &=
  \left| \frac 1 {nT}
  {\rm Tr}\left[
   \left( \widehat{X^\perp} - X^\perp \right)'
    \left( \widehat{X^\perp} + X^\perp \right)
  \right] 
  \right|^2
 \\
 &\leq \underbrace{\frac 1 {nT}\left\| \widehat{X^\perp} - X^\perp \right\|_2^2}_{=o_P(1)}
 \;
  \underbrace{\frac 1 {nT}\left\| \widehat{X^\perp} + X^\perp \right\|_2^2
  }_{\leq \frac 1 {nT}\left\| X^\perp \right\|_2^2 + o_P(1) = \largeO_P(1)}
  \\
  &= o_P(1),
\end{align*}
and therefore 
\begin{align*}
   \frac 1 {nT} \left\|  \widehat{X^\perp} \right\|_2^2 
   &=  \frac 1 {nT} \left\| X^\perp \right\|_2^2 + o_P(1)
 \\
 &= c+ o_P(1) ,
\end{align*}
where we also used assumption (ii) in the theorem.
Similarly,
\begin{align*}
  \left| \frac 1 {nT} {\rm Tr}\left( Y' \widehat{X^\perp} \right) 
   - \frac 1 {nT} {\rm Tr}\left( Y'  X^\perp \right) 
   \right|^2
  &= \left| \frac 1 {nT} {\rm Tr}\left[ Y'  \left( \widehat{X^\perp} - X^\perp \right) \right]  \right|^2
  \\
  &\leq  \underbrace{\frac 1 {nT} \| Y \|^2_2}_{=\largeO_P(1)}
        \underbrace{\frac 1 {nT}\left\| \widehat{X^\perp} - X^\perp \right\|_2^2}_{=o_P(1)}
        \\
     &= o_P(1) ,
\end{align*}
where we used assumption (i) of the theorem.
Combining the last two results we find that
\begin{align*}
   \widehat \beta_{\rm PC} &:=   
      \frac{  \sum_{i=1}^n \sum_{t=1}^T    Y_{it} \widehat{X^\perp_{it}}} {\sum_{i=1}^n \sum_{t=1}^T  (\widehat{X^\perp_{it}})^2 }
      = \frac{\frac 1 {nT} {\rm Tr}\left( Y' \widehat{X^\perp} \right) }
         {\frac 1 {nT} \left\|  \widehat{X^\perp} \right\|_2^2}
       = \frac{\frac 1 {nT} {\rm Tr}\left( Y'  X^\perp  \right) }
         {\frac 1 {nT} \left\|  X^\perp \right\|_2^2 + o_P(1)} +o_P(1)
     \\
     &= \beta^* +\frac{\frac 1 {nT} {\rm Tr}\left( E'  X^\perp  \right)}
         {\frac 1 {nT} \left\|  X^\perp \right\|_2^2}+ o_P(1) ,
\end{align*}
where in the last step we used the fact that $\frac 1 {nT} \left\|  X^\perp \right\|_2^2\to_{P}c>0$, and that $\frac 1 {nT} {\rm Tr}\left( E'  X^\perp  \right)=o_{P}(1)$ by assumption.
This concludes the proof.\\
\end{proof}

\begin{proof}[\bf Proof of Theorem~\ref{thm::IFE_consistencyNEW}]
Define
\begin{align*}
  E^\parallel_R &:= \argmin_{\big\{ G \in \mathbb{R}^{n \times T}: {\rm rank}(G) \leq R \big\}}  \left\|  E^\parallel -  G \right\|^2_2 ,
  &
  E^+_R &:= E^\parallel - E^\parallel_R.
\end{align*}
Remember that $\widehat \beta_{\rm IFE}  
       = \argmin_{\beta \in \mathbb{R}}  Q(\beta)$, where
\begin{align*} 
      Q(\beta)
      &:= \min_{\big\{ G \in \mathbb{R}^{n \times T}: {\rm rank}(G) \leq R  \big\}}
      \left\| Y - X \cdot \beta - G \right\|_2^2
    \\
     &= \min_{\big\{ G \in \mathbb{R}^{n \times T}: {\rm rank}(G) \leq R  \big\}}
      \left\| (\Delta \beta)  X + E - G \right\|_2^2
    \\  
     &= \min_{\big\{ G \in \mathbb{R}^{n \times T}: {\rm rank}(G) \leq R  \big\}}
      \left\| (\Delta \beta)  X + E^\perp + E^\parallel - G \right\|_2^2
    \\  
     &= \min_{\big\{ G \in \mathbb{R}^{n \times T}: {\rm rank}(G) \leq R  \big\}}
      \left\| (\Delta \beta)  X  + \left( E^\perp + E^+_R  \right) +  E^\parallel_R - G \right\|_2^2
    \\
     &\geq  \min_{\big\{ \tilde G \in \mathbb{R}^{n \times T}: {\rm rank}(\tilde G) \leq 2R  \big\}}
      \left\| (\Delta \beta)  X  + \left( E^\perp + E^+_R  \right) -   \tilde G \right\|_2^2
    \\  
      &=:  \underline Q(\beta) \, .
\end{align*}
Since we can choose $\beta^*$ and $G=E^\parallel_R$ in the above minimization problems we have
\begin{align}
 \underline Q(\widehat \beta_{\rm IFE} )
  \leq Q(\widehat \beta_{\rm IFE} )
  \leq Q(\beta^*) \leq 
  \left\| Y - X \cdot \beta^* - E^\parallel_R \right\|_2^2
 = \| E^\perp +E^+_R \|^2_2 .
 \label{UpperBoundUnderlineQ}
\end{align}
Our goal is to find a lower bound for $\underline Q(\beta) $ that allows us to conclude consistency of $\widehat \beta_{\rm IFE}$.  
Write any rank-\(2R\) matrix as \(\tilde G=\lambda f'\), let \(P_f\) be the orthogonal projector onto the span of $f$, and let $M_f=I-P_f$. Then
\[
\underline Q(\beta)
=
\min_{f}
\mathrm{Tr}\!\big[(\Delta\beta\cdot X + E^\perp+E^+_R) M_f (\Delta\beta\cdot X + E^\perp+E^+_R)'\big].
\]
Expanding gives, wpa1,
\begin{align}
\underline Q(\beta)
&\geq
\mathrm{Tr}\!\big[(\Delta\beta \, X)M_f(\Delta\beta \, X)'\big]
+ \|E^\perp+E^+_R\|_2^2
\nonumber \\
&\quad
- \mathrm{Tr}\!\big[(E^\perp+E^+_R)P_f(E^\perp+E^+_R)'\big]
\nonumber \\ &\quad
+ 2\,\mathrm{Tr}\!\big[(\Delta\beta \, X)(E^\perp)'\big] -2\big|\mathrm{Tr}\!\big[(\Delta\beta  X)P_f(E^\perp)'\big]\big|
\nonumber \\ &\quad
+ 2\,\mathrm{Tr}\!\big[(\Delta\beta \, X) M_f (E^+_R)'\big] 
\nonumber \\
&\geq c\,nT\, |\Delta \beta|^2 + \|E^\perp+E^+_R\|_2^2 
\nonumber \\ & \quad
 - 2R\|E^\perp\|^2 
 - \|E^+_R\|_2^2 
 - 2\sqrt{2R}\,\|E^\perp\|\,\|E^+_R\|_2
\nonumber \\ & \quad 
+ 2\,\mathrm{Tr}\!\big[(\Delta\beta \, X)(E^\perp)'\big]
- 2\sqrt{2R}\,|\Delta\beta| \,\|X\|_2\,\|E^\perp\|
\nonumber \\ & \quad 
- 2 |\Delta\beta| \|X\|_2 \|E^+_R\|_2
\nonumber \\
&\geq c\,nT\, |\Delta \beta|^2 + \|E^\perp+E^+_R\|_2^2 
   + o_P(nT) + o_P( nT  |\Delta \beta|) ,
 \label{LowerBoundUnderlineQ}  
\end{align}
where we used that
$$
  \left|\mathrm{Tr}\!\big[(\Delta\beta \, X) M_f (E^+_R)'\big]
  \right|
   \leq 2 |\Delta\beta| \left\|  X M_f \right\|_2  \left\| E^+_R \right\|
  \leq  2 |\Delta\beta| \left\|  X \right\|_2  \left\| E^+_R \right\| .
$$
For the last step, notice that $\left\|  X M_f \right\|_2 \leq \left\|  X \right\|_2$, because
$$
  \left\| X  M_f \right\|_2^2
   = {\rm Tr}[ X  M_f   X' ] 
   = {\rm Tr}[ X    X' ] -  {\rm Tr}[ X  P_f  X' ]
   = \left\|  X \right\|_2^2 - \underbrace{\left\| X  P_f \right\|_2^2}_{\geq 0} .
$$
In \eqref{LowerBoundUnderlineQ} 
we also used assumption (i)--(iv) to control each remainder term.  In particular:
\begin{itemize}
  \item \( \displaystyle 2R\|E^\perp\|^2 = o_P(nT)\) and \(\sqrt{2R}\,\|E^\perp\|\,\|E^+_R\|_2 = o_P(nT)\) by (iii) together with \(\|E^+_R\|_2 = o_P(\sqrt{nT})\) from (iv).
  \item \( \|E^+_R\|_2^2 = o_P(nT)\) by (iv).
  \item \( \mathrm{Tr}((\Delta\beta\,X)E^\perp{}')=\Delta\beta\;\mathrm{Tr}(X E^\perp{}')=o_P(nT|\Delta\beta|)\) by (i).
  \item \( \sqrt{2R}\,|\Delta\beta|\,\|X\|_2\,\|E^\perp\| = |\Delta\beta|\,o_P(nT)\) and \( |\Delta\beta|\,\|X\|_2\,\|E^+_R\|_2 = |\Delta\beta|\,o_P(nT)\) because \(\|X\|_2=\largeO_P(\sqrt{nT})\), \(\|E^\perp\|=o_P(\sqrt{nT/R})\) and \(\|E^+_R\|_2=o_P(\sqrt{nT})\).
\end{itemize}
Combining these gives the remainder \(o_P(nT)+o_P(nT|\Delta\beta|)=o_P\big(nT(1+|\Delta\beta|)\big)\), which yields \eqref{LowerBoundUnderlineQ}.
Combining \eqref{UpperBoundUnderlineQ} and \eqref{LowerBoundUnderlineQ} we find that
\[
c\,|\widehat\beta_{\rm IFE}-\beta^*|^2 \le o_P\big(1+|\widehat\beta_{\rm IFE}-\beta^*|\big),
\]
which implies $\widehat\beta_{\rm IFE} = \beta^* + o_P(1)$.
\end{proof}

\subsection{Proof for Section~\ref{ssection::lowlevel}}

Before proving Theorem~\ref{thm::lowlevel}, it is useful to establish the following intermediate lemma.

\begin{lemma}[\bf Medium-level sufficient conditions]
\label{lemma::medium}
Let $U_i \in \mathcal{U}$ and $V_t \in \mathcal{V}$ be random vectors, and define $U = (U_1, \ldots, U_n)$ and $V = (V_1, \ldots, V_T)$. Assume that $(Y_{it}, X_{it}, U_i, V_t)$ are identically distributed over $(i,t)$, and that $(Y_{it}, X_{it})$ have uniformly bounded fourth moments conditional on $(U, V)$. Define $X^\parallel$, $X^\perp$, $\beta^*$, $E$, $E^\parallel$, and $E^\perp$ as in Theorems~\ref{thm::PC_consistencyNEW} and \ref{thm::IFE_consistencyNEW}. Assume:
\begin{enumerate}[(i)]
    \item $\mathbb{E}[(X_{it}^\perp)^2] > 0$.
    \item $(U_i : i = 1, \ldots, n)$ are i.i.d., and $(V_t : t = 1, \ldots, T)$ is strictly stationary and $\alpha$-mixing with summable mixing coefficients.
    \item Conditional on $(U, V)$, the process $\{(X_{it}^\perp, E_{it}^\perp) : t = 1, \ldots, T\}$ is $\alpha$-mixing with summable mixing coefficients for each $i$, and the vectors $\{(X_{it}^\perp, E_{it}^\perp)\}_{t=1}^T$ are independent across $i$.
    \item $\mathbb{E}(X_{it} \,|\, U_i, V_t) = \mathbb{E}(X_{it} \,|\, U, V)$ and $\mathbb{E}(Y_{it} \,|\, U_i, V_t) = \mathbb{E}(Y_{it} \,|\, U, V)$.
    \item The conditional mean function $h_X(u, v) := \mathbb{E}(X_{it} \,|\, U_i = u, V_t = v)$ admits a singular value decomposition
    \begin{align*}
        h_X(u, v) = \sum_{r=1}^\infty \sigma_r \phi_r(u) \psi_r(v),
    \end{align*}
    where $\sigma_1 \geq \sigma_2 \geq \ldots \geq 0$, and $\mathbb{E}[\phi_r(U_i) \phi_s(U_i)] = \mathbbm{1}(r = s)$, and $\mathbb{E}[\psi_r(V_t) \psi_s(V_t)] = \mathbbm{1}(r = s)$, and
    \begin{align*}
        \sum_{r > R_{nT}} \sigma_r^2 = o(1).
    \end{align*}
    The same type of decomposition holds for $h_E(u, v) := \mathbb{E}(E_{it} \,|\, U_i = u, V_t = v)$ with singular values $\tilde{\sigma}_r$ satisfying $\sum_{r > R_{nT}} \tilde{\sigma}_r^2 = o(1)$.
    \item $R_{nT} = o\left(\min\left(\sqrt{T}, \sqrt{n} / (\log T)^2\right)\right)$.
\end{enumerate}
Then the assumptions of Theorems~\ref{thm::PC_consistencyNEW} and \ref{thm::IFE_consistencyNEW} are satisfied, and hence $\widehat{\beta}_{\rm PC} = \beta^* + o_P(1)$ and $\widehat{\beta}_{\rm IFE} = \beta^* + o_P(1)$.
\end{lemma}

\begin{proof}[\bf Proof of Lemma~\ref{lemma::medium}]
Throughout we write $R = R_{nT}$, use $\|\cdot\|$ for the spectral norm and 
$\|\cdot\|_2$ for the Frobenius norm, and set
$X^\parallel_{it} = h_X(U_i,V_t)$, $X^\perp_{it} = X_{it} - X^\parallel_{it}$,
and analogously for $E^\parallel_{it}$ and $E^\perp_{it}$.
We verify each condition of Theorems~\ref{thm::PC_consistencyNEW} 
and~\ref{thm::IFE_consistencyNEW} in turn.

\medskip\noindent\textbf{Score conditions.}
By definition of $\beta^*$ we have $\mathbb{E}[X^\perp_{it} E_{it}] = 0$.
For fixed $i$, the process $\{X^\perp_{it} E_{it}\}$ is strictly stationary and
$\alpha$-mixing with summable coefficients by assumption (iii), and has a bounded
second moment by the fourth-moment bound. A standard mixing law of large numbers
applied row-by-row and then averaged across independent rows gives
\[
\frac{1}{nT}\sum_{i,t} X^\perp_{it} E_{it} = o_P(1).
\]
For the second score condition, write
\[
\frac{1}{nT}\sum_{i,t} X_{it} E^\perp_{it}
= \frac{1}{nT}\sum_{i,t} X^\parallel_{it} E^\perp_{it}
+ \frac{1}{nT}\sum_{i,t} X^\perp_{it} E^\perp_{it}.
\]
For the first term, conditional on $(U,V)$, $X^\parallel_{it}$ is deterministic and
$\mathbb{E}[E^\perp_{it} \mid U_i, V_t] = 0$, so the conditional mean of the sum
is zero. Its conditional variance is $\largeO_P((nT)^{-1})$ by independence across $i$
and bounded moments, hence the term is $o_P(1)$.
For the second term, $\mathbb{E}[X^\perp_{it} E^\perp_{it}] = \mathbb{E}[X^\perp_{it} E_{it}]
- \mathbb{E}[X^\perp_{it} E^\parallel_{it}] = 0$, where the second equality uses
$\mathbb{E}[X^\perp_{it} \mid U_i, V_t] = 0$ and the fact that $E^\parallel_{it}$ is
a function of $(U_i, V_t)$. The same mixing law of large numbers gives
$\frac{1}{nT}\sum_{i,t} X^\perp_{it} E^\perp_{it} = o_P(1)$.
Thus $\frac{1}{nT}\sum_{i,t} X_{it} E^\perp_{it} = o_P(1)$.

\medskip\noindent\textbf{Moment bounds and non-degeneracy.}
Uniformly bounded fourth moments together with the sampling scheme in assumption (ii)
imply $\frac{1}{nT}\sum_{i,t} X_{it}^2 = \largeO_P(1)$ and
$\frac{1}{nT}\sum_{i,t} Y_{it}^2 = \largeO_P(1)$ by a law of large numbers for double
arrays. Assumption (i) gives
\[
\frac{1}{nT}\sum_{i,t}(X^\perp_{it})^2 \to_P \mathbb{E}[(X^\perp_{it})^2] =: c > 0.
\]

\medskip\noindent\textbf{Spectral norm bounds.}
Lemma B.3 of \citet{wang2022lowrankpanelquantileregression}, applied using the
conditional mixing and fourth-moment conditions in assumption (iii), gives
\[
\|X^\perp\| = \largeO_P(\sqrt{n} + \sqrt{T \log^2 T}),
\qquad
\|E^\perp\| = \largeO_P(\sqrt{n} + \sqrt{T \log^2 T}).
\]
Assumption (vi) requires $R = o(\sqrt{T})$ and $R = o(\sqrt{n}/(\log T)^2)$,
which implies $R/T \to 0$ and $R \log^2 T / n \to 0$. Therefore
\[
\frac{R}{nT}\|X^\perp\|^2
= \largeO_P\!\left(R\left(\frac{1}{T} + \frac{\log^2 T}{n}\right)\right) = o_P(1),
\]
and the same bound holds for $\|E^\perp\|^2$.

\medskip\noindent\textbf{Low-rank approximations.}
By the functional SVD in assumption (v), $h_X(u,v) = \sum_{r=1}^\infty \sigma_r
\phi_r(u)\psi_r(v)$ with orthonormal singular functions. Under the i.i.d.\
sampling of $(U_i)$ and the stationarity and mixing of $(V_t)$ in assumption (ii), 
and the orthonormality $\mathbb{E}[\phi_r(U_i)\phi_s(U_i)] = \mathbbm{1}(r=s)$, 
$\mathbb{E}[\psi_r(V_t)\psi_s(V_t)] = \mathbbm{1}(r=s)$ from assumption (v),
taking expectations yields
\[
\mathbb{E}\Bigl[\tfrac{1}{nT}\|X^\parallel - G_{X,R}\|_2^2\Bigr]
\leq \sum_{r > R}\sigma_r^2 = o(1),
\]
where $G_{X,R}$ denotes the best rank-$R$ Frobenius approximation to $X^\parallel$.
Markov's inequality then gives $\tfrac{1}{nT}\|X^\parallel - G_{X,R}\|_2^2 = o_P(1)$.
The same argument applied to $h_E$ gives 
$\min_{\mathrm{rank}(G) \leq R} \tfrac{1}{nT}\|E^\parallel - G\|_2^2 = o_P(1)$.

\medskip\noindent\textbf{Non-collinearity for the IFE estimator.}
For any $G$ with $\mathrm{rank}(G) \leq 2R$,
\[
\frac{1}{nT}\|X - G\|_2^2
= \frac{1}{nT}\|X^\perp\|_2^2
+ \frac{1}{nT}\|X^\parallel - G\|_2^2
+ \frac{2}{nT}\,\mathrm{Tr}\!\left[(X^\parallel - G)'X^\perp\right].
\]
The first term converges to $c > 0$ by the non-degeneracy result above, and the
second term is non-negative. To control the cross-term, let $G_{X,R}$ denote the
best rank-$R$ Frobenius approximation to $X^\parallel$, and split
\[
\mathrm{Tr}\!\left[(X^\parallel - G)'X^\perp\right]
= \mathrm{Tr}\!\left[(X^\parallel - G_{X,R})'X^\perp\right]
+ \mathrm{Tr}\!\left[(G_{X,R} - G)'X^\perp\right].
\]
For the first piece, Cauchy--Schwarz in Frobenius norm and the low-rank
approximation result above give
\[
\frac{2}{nT}\left|\mathrm{Tr}\!\left[(X^\parallel - G_{X,R})'X^\perp\right]\right|
\leq 2\sqrt{\tfrac{1}{nT}\|X^\parallel - G_{X,R}\|_2^2}\,
     \sqrt{\tfrac{1}{nT}\|X^\perp\|_2^2}
= o_P(1) \cdot \largeO_P(1) = o_P(1).
\]
For the second piece, $\mathrm{rank}(G_{X,R} - G) \leq 3R$ implies
$\|G_{X,R} - G\|_1 \leq \sqrt{3R}\,\|G_{X,R} - G\|_2$, so the nuclear--spectral
inequality $|\mathrm{Tr}[A'B]| \leq \|A\|_1 \|B\|$ gives
\[
\frac{2}{nT}\left|\mathrm{Tr}\!\left[(G_{X,R} - G)'X^\perp\right]\right|
\leq \frac{2\sqrt{3R}}{nT}\,\|G_{X,R} - G\|_2\,\|X^\perp\|.
\]
Using the triangle inequality $\|G_{X,R} - G\|_2 \leq \|X^\parallel - G_{X,R}\|_2
+ \|X^\parallel - G\|_2$ and $2ab \leq \delta a^2 + \delta^{-1} b^2$ for
$\delta \in (0,1)$, this is further bounded by
\[
2\sqrt{\tfrac{\|X^\parallel - G_{X,R}\|_2^2}{nT}}
 \sqrt{\tfrac{3R\,\|X^\perp\|^2}{nT}}
+ \delta\,\tfrac{\|X^\parallel - G\|_2^2}{nT}
+ \tfrac{3R\,\|X^\perp\|^2}{\delta\,nT}
= \delta\,\tfrac{1}{nT}\|X^\parallel - G\|_2^2 + o_P(1),
\]
where the remainder is $o_P(1)$ by the low-rank approximation result, the
spectral norm bound, and assumption (vi), and is uniform in $G$. Combining,
\[
\frac{1}{nT}\|X - G\|_2^2
\geq \frac{1}{nT}\|X^\perp\|_2^2
   + (1 - \delta)\,\frac{1}{nT}\|X^\parallel - G\|_2^2
   + o_P(1).
\]
Since $\delta < 1$ and the middle term is non-negative,
\[
\min_{\mathrm{rank}(G) \leq 2R} \frac{1}{nT}\|X - G\|_2^2 \geq c + o_P(1).
\]

\medskip\noindent\textbf{Conclusion.}
All conditions of Theorems~\ref{thm::PC_consistencyNEW} and~\ref{thm::IFE_consistencyNEW}
have been verified under assumptions (i)--(vi) of the lemma. Therefore
\[
\widehat{\beta}_{\mathrm{PC}} = \beta^* + o_P(1),
\qquad
\widehat{\beta}_{\mathrm{IFE}} = \beta^* + o_P(1).
\]
\end{proof}

\begin{proof}[\bf Proof of Theorem~\ref{thm::lowlevel}]
Let $Z_{it}=g(U_i,V_t,\varepsilon_{it})=(Y_{it},X_{it})'$ and define
\[
h_X(u,v):=\mathbb{E}[X_{it}\mid U_i=u,V_t=v],
\qquad 
h_E(u,v):=\mathbb{E}[E_{it}\mid U_i=u,V_t=v],
\]
with $X^\parallel_{it}=h_X(U_i,V_t)$, $X^\perp_{it}=X_{it}-X^\parallel_{it}$,
$E_{it}=Y_{it}-\beta^* X_{it}$, $E^\parallel_{it}=h_E(U_i,V_t)$,
and $E^\perp_{it}=E_{it}-E^\parallel_{it}$ as in Lemma~\ref{lemma::medium}.
We verify assumptions (i)--(vi) of Lemma~\ref{lemma::medium}.

\medskip\noindent\textbf{(i) Non-degeneracy: $\mathbb{E}[(X^\perp_{it})^2]>0$.}

By assumption (iii) of the current theorem, 
$\mathrm{Var}(X_{it}\mid U_i,V_t)>0$ almost surely.
Since $\varepsilon_{it}$ is independent of $(U,V)$ by assumption (ii), 
we have
\[
\mathrm{Var}(X_{it}\mid U_i,V_t) 
= \mathbb{E}\big[(X_{it}-h_X(U_i,V_t))^2 \mid U_i,V_t\big]
= \mathbb{E}\big[(X^\perp_{it})^2 \mid U_i,V_t\big].
\]
Taking expectations,
\[
\mathbb{E}\big[(X^\perp_{it})^2\big] 
= \mathbb{E}\big[\mathrm{Var}(X_{it}\mid U_i,V_t)\big] > 0,
\]
which verifies assumption (i) of Lemma~\ref{lemma::medium}.

\medskip\noindent\textbf{(ii) Sampling of $(U,V)$: $U_i$ i.i.d., $V_t$ stationary mixing.}

This is directly assumed in assumption (i) of the current theorem:
$(U_i : i \geq 1)$ are i.i.d., and $(V_t : t \geq 1)$ is strictly stationary 
and $\alpha$-mixing with summable coefficients.

\medskip\noindent\textbf{(iii) Conditional mixing and independence of $(X^\perp_{it}, E^\perp_{it})$.}

We must show that conditional on $(U,V)$:
\begin{itemize}
    \item[(a)] $\{(X^\perp_{it}, E^\perp_{it}) : t=1,\ldots,T\}$ is $\alpha$-mixing 
    with summable coefficients for each $i$;
    \item[(b)] $\{(X^\perp_{it}, E^\perp_{it})\}_{t=1}^T$ are independent across $i$.
\end{itemize}

Since $\varepsilon_{it}$ is independent of $(U,V)$ by assumption (ii), 
conditioning on $(U,V)$ does not alter the distribution of $\varepsilon_{it}$.
The DGP gives
\[
X^\perp_{it} = g_X(U_i, V_t, \varepsilon_{it}) - h_X(U_i, V_t),
\]
where $g_X$ denotes the second component of $g$.
Conditional on $(U,V)$, the terms $U_i$, $V_t$, and $h_X(U_i,V_t)$ are deterministic,
so the only randomness in $X^\perp_{it}$ comes from $\varepsilon_{it}$.

\emph{Part (a): Mixing in $t$.}
By assumption (ii), for each fixed $i$, the process 
$\{\varepsilon_{it} : t \geq 1\}$ is strictly stationary and $\alpha$-mixing 
with summable coefficients.
Since $g$ is measurable and $(U_i, V_t)$ is fixed when conditioning on $(U,V)$,
the process $\{X^\perp_{it} : t \geq 1\}$ inherits the $\alpha$-mixing property
(mixing is preserved under measurable transformations).
The same argument applies to $E^\perp_{it}$, and hence to the joint process
$\{(X^\perp_{it}, E^\perp_{it}) : t \geq 1\}$.

\emph{Part (b): Independence across $i$.}
By assumption (ii), the rows $\{\varepsilon_{it} : t \geq 1\}$ are independent 
across $i$ as stochastic processes.
Since $X^\perp_{it} = g_X(U_i, V_t, \varepsilon_{it}) - h_X(U_i, V_t)$ depends 
on $\varepsilon_{i\cdot} = \{\varepsilon_{is} : s \geq 1\}$ only,
and conditional on $(U,V)$ the function $g_X(U_i, V_t, \cdot) - h_X(U_i, V_t)$ 
is deterministic, the processes $\{(X^\perp_{it}, E^\perp_{it})\}_{t=1}^T$ 
inherit independence across $i$ from the row-independence of $\varepsilon$.

This verifies assumption (iii) of Lemma~\ref{lemma::medium}.

\medskip\noindent\textbf{(iv) Locality: $\mathbb{E}[X_{it} \mid U_i, V_t] = \mathbb{E}[X_{it} \mid U, V]$.}

Since $Z_{it} = g(U_i, V_t, \varepsilon_{it})$ and $\varepsilon_{it}$ is independent of $(U,V)$,
we have
\[
\mathbb{E}[X_{it} \mid U, V] 
= \mathbb{E}[g_X(U_i, V_t, \varepsilon_{it}) \mid U, V]
= \int g_X(U_i, V_t, e) \, dP_\varepsilon(e),
\]
where $P_\varepsilon$ is the marginal distribution of $\varepsilon_{it}$.
The right-hand side depends on $(U,V)$ only through $(U_i, V_t)$, so
\[
\mathbb{E}[X_{it} \mid U, V] = \mathbb{E}[X_{it} \mid U_i, V_t] = h_X(U_i, V_t).
\]
The same argument applies to $Y_{it}$ and hence to $E_{it}$.
This verifies assumption (iv) of Lemma~\ref{lemma::medium}.

\medskip\noindent\textbf{(v) SVD tail conditions from smoothness.}

Assumption (iv) of the current theorem states that $U_i \in \mathbb{R}^{d_U}$ and $V_t \in \mathbb{R}^{d_V}$ for fixed integers $d_U, d_V \geq 1$, and that $h_X(u,v)$ is $s$-times continuously differentiable with $s > \max(d_U, d_V)/2$. We must show this implies the SVD tail condition in Lemma~\ref{lemma::medium}(v).

The function $h_X(u,v)$ defines an integral operator $T_X : L^2(\mathcal{V}) \to L^2(\mathcal{U})$ via
\[
(T_X \psi)(u) = \int h_X(u,v) \psi(v) \, dP_V(v),
\]
where $P_V$ is the distribution of $V_t$. The singular values $\sigma_1 \geq \sigma_2 \geq \ldots$ of this operator coincide with those in the functional SVD
\[
h_X(u,v) = \sum_{r=1}^\infty \sigma_r \phi_r(u) \psi_r(v).
\]

Classical results in approximation theory establish that for $s$-smooth kernels on domains in $\mathbb{R}^d$, the singular values decay as $\sigma_r = \largeO(r^{-s/d})$ where $d = \max(d_U, d_V)$. See \citet{birman1977estimates} for the foundational results on spectral asymptotics of integral operators with smooth kernels.

Since $s > d/2$ by assumption, we have $2s/d > 1$, and therefore
\[
\sum_{r > R} \sigma_r^2 = O\left(\sum_{r > R} r^{-2s/d}\right) = \largeO(R^{1 - 2s/d}) = o(1)
\]
as $R \to \infty$. Since $R_{nT} \to \infty$ by assumption (v) of the current theorem, we have $\sum_{r > R_{nT}} \sigma_r^2 = o(1)$.

The same argument applies to $h_E(u,v) = \mathbb{E}[E_{it} \mid U_i = u, V_t = v]$, which is also $s$-times continuously differentiable (since $E_{it} = Y_{it} - \beta^* X_{it}$ and both $h_Y$ and $h_X$ satisfy the smoothness condition).

This verifies assumption (v) of Lemma~\ref{lemma::medium}.

\medskip\noindent\textbf{(vi) Growth rate of $R_{nT}$.}

Assumption (v) of the current theorem directly assumes $R_{nT} \to \infty$ and
\[
R_{nT} = o\big(\min(\sqrt{T}, \sqrt{n}/(\log T)^2)\big),
\]
which implies assumption (vi) of Lemma~\ref{lemma::medium}.

\medskip\noindent\textbf{Moment bounds: uniformly bounded fourth moments conditional on $(U,V)$.}

Assumption (iii) of the current theorem gives the uniform-in-$(u,v)$ bound
\[
\sup_{u \in \mathcal{U}, v \in \mathcal{V}} 
\mathbb{E}\big[\|g(u, v, \varepsilon)\|^4\big] =: M < \infty.
\]
Since $\varepsilon_{it}$ is independent of $(U,V)$, we have for all $(i,t)$:
\[
\mathbb{E}[\|Z_{it}\|^4 \mid U, V] 
= \mathbb{E}[\|g(U_i, V_t, \varepsilon_{it})\|^4 \mid U, V]
= \mathbb{E}[\|g(U_i, V_t, \varepsilon)\|^4] \leq M
\quad \text{a.s.}
\]
This provides the uniform fourth-moment bound required in 
Lemma~\ref{lemma::medium}.

\medskip\noindent\textbf{Identical distribution.}

The tuple $(Y_{it}, X_{it}, U_i, V_t)$ is identically distributed over $(i,t)$ because:
\begin{itemize}
    \item $(U_i)$ are i.i.d.\ by assumption (i);
    \item $(V_t)$ is strictly stationary by assumption (i);
    \item $(\varepsilon_{it})$ has the same marginal distribution for all $(i,t)$ 
    by assumption (ii) (rows are i.d.\ and each row is stationary in $t$);
    \item $g$ does not depend on $(i,t)$.
\end{itemize}
This verifies the identical distribution requirement in 
Lemma~\ref{lemma::medium}.

\medskip\noindent\textbf{Conclusion.}
All assumptions (i)--(vi) and the auxiliary conditions (bounded fourth moments, 
identical distribution) of Lemma~\ref{lemma::medium} have been 
verified under assumptions (i)--(v) of the current theorem.
Therefore, the conclusions of Lemma~\ref{lemma::medium} hold,
and hence
\[
\widehat{\beta}_{\mathrm{PC}} = \beta^* + o_P(1),
\qquad
\widehat{\beta}_{\mathrm{IFE}} = \beta^* + o_P(1).
\]
\end{proof}

\subsection{Proofs for Section~\ref{ssection::conditional}}

\begin{proof}[\bf Proof of Theorem~\ref{thm::conditional}]
We verify the conditions of Theorems~\ref{thm::PC_consistencyNEW} 
and~\ref{thm::IFE_consistencyNEW} with $\beta^*$ replaced by $\beta^*_{nT}$, 
where all probability statements are conditional on $(U, V)$.

Define $E_{it} := Y_{it} - \beta^*_{nT} X_{it}$, $E_{it}^\parallel := 
\mathbb{E}(E_{it} \,|\, U_i, V_t)$, and $E_{it}^\perp := E_{it} - E_{it}^\parallel$. 
We write $R = R_{nT}$ throughout.

\medskip\noindent\textbf{Score conditions.}
By construction of $\beta^*_{nT}$,
\[
\frac{1}{nT}\sum_{i,t} \mathbb{E}(X_{it}^\perp E_{it} \,|\, U_i, V_t)
= \frac{1}{nT}\sum_{i,t} {\rm Cov}(X_{it}, Y_{it} \,|\, U_i, V_t) 
- \beta^*_{nT} \frac{1}{nT}\sum_{i,t} {\rm Var}(X_{it} \,|\, U_i, V_t) = 0.
\]
Therefore
\[
\frac{1}{nT}\sum_{i,t} X_{it}^\perp E_{it}
= \frac{1}{nT}\sum_{i,t} \left(X_{it}^\perp E_{it} 
  - \mathbb{E}[X_{it}^\perp E_{it} \,|\, U_i, V_t]\right).
\]
Conditional on $(U, V)$, the summands on the right have mean zero. For each $i$, 
the process $\{X_{it}^\perp E_{it}\}$ is $\alpha$-mixing with uniformly summable 
coefficients (inherited from $\varepsilon_{it}$ via measurable transformation), and 
the rows are independent across $i$ by assumption (i). The uniform fourth moment 
bound in assumption (ii) implies $\sup_{i,t} \mathbb{E}[(X_{it}^\perp E_{it})^2 
\,|\, U_i, V_t] < \infty$. A triangular-array law of large numbers for row-wise 
independent $\alpha$-mixing sequences with uniformly summable mixing coefficients 
(see e.g.\ \citealt{rio2017asymptotic}) therefore gives
\[
\frac{1}{nT}\sum_{i,t} X_{it}^\perp E_{it} = o_P(1).
\]
For the second score condition, write
\[
\frac{1}{nT}\sum_{i,t} X_{it} E_{it}^\perp 
= \frac{1}{nT}\sum_{i,t} X_{it}^\parallel E_{it}^\perp 
+ \frac{1}{nT}\sum_{i,t} X_{it}^\perp E_{it}^\perp.
\]
For the first term, conditional on $(U,V)$, $X_{it}^\parallel$ is deterministic and 
$\mathbb{E}[E_{it}^\perp \,|\, U_i, V_t] = 0$. The conditional variance of the sum 
is $\largeO((nT)^{-1})$ by independence across $i$ and bounded moments, so this term is 
$o_P(1)$. For the second term, $\mathbb{E}(X_{it}^\perp E_{it}^\perp \,|\, U_i, V_t) 
= \mathbb{E}(X_{it}^\perp E_{it} \,|\, U_i, V_t) - \mathbb{E}(X_{it}^\perp 
E_{it}^\parallel \,|\, U_i, V_t) = \mathbb{E}(X_{it}^\perp E_{it} \,|\, U_i, V_t)$, 
where the last equality uses $\mathbb{E}(X_{it}^\perp \,|\, U_i, V_t) = 0$ and the 
fact that $E_{it}^\parallel$ is a function of $(U_i, V_t)$. The same decomposition 
and triangular-array LLN argument then give $\frac{1}{nT}\sum_{i,t} X_{it}^\perp 
E_{it}^\perp = o_P(1)$. Thus $\frac{1}{nT}\sum_{i,t} X_{it} E_{it}^\perp = o_P(1)$.

\medskip\noindent\textbf{Moment bounds and non-degeneracy.}
Assumption (ii) provides the uniform fourth moment bound. By the same 
triangular-array LLN,
\[
\frac{1}{nT}\sum_{i,t} X_{it}^2 = \largeO_P(1), \qquad 
\frac{1}{nT}\sum_{i,t} Y_{it}^2 = \largeO_P(1).
\]
For non-degeneracy, the triangular-array LLN gives
\[
\frac{1}{nT}\sum_{i,t} (X_{it}^\perp)^2 \to_P 
\frac{1}{nT}\sum_{i,t} \mathbb{E}[(X_{it}^\perp)^2 \,|\, U_i, V_t] 
= \frac{1}{nT}\sum_{i,t} {\rm Var}(X_{it} \,|\, U_i, V_t) \geq c,
\]
where the inequality is assumption (ii).

\medskip\noindent\textbf{Spectral norm bounds.}
Conditional on $(U, V)$, the entries of $X^\perp$ satisfy $X_{it}^\perp = 
g_X(U_i, V_t, \varepsilon_{it}) - h_X(U_i, V_t)$ where $h_X(U_i, V_t)$ is 
deterministic. The rows are independent across $i$ and each row is $\alpha$-mixing 
with uniformly summable coefficients by assumption (i). Lemma~B.3 of 
\citet{wang2022lowrankpanelquantileregression} therefore gives
\[
\|X^\perp\| = \largeO_P(\sqrt{n} + \sqrt{T \log^2 T}),
\]
and the same bound holds for $E^\perp$. Combined with assumption (iv),
\[
\frac{R}{nT}\|X^\perp\|^2 
= \largeO_P\!\left(R\left(\frac{1}{T} + \frac{\log^2 T}{n}\right)\right) = o_P(1),
\]
and similarly for $E^\perp$.

\medskip\noindent\textbf{Non-collinearity for the IFE estimator.}
For any $G$ with $\mathrm{rank}(G) \leq 2R$, expand
\[
\frac{1}{nT}\|X - G\|_2^2 = \frac{1}{nT}\|X^\perp\|_2^2 
+ \frac{1}{nT}\|X^\parallel - G\|_2^2 
+ \frac{2}{nT}\,\mathrm{Tr}\!\left[(X^\parallel - G)'X^\perp\right].
\]
The first term satisfies $\frac{1}{nT}\|X^\perp\|_2^2 \geq c + o_P(1)$ by the 
non-degeneracy result above. The second term is non-negative and can be dropped. 
For the cross-term, applying the nuclear--spectral inequality and Young's inequality 
as in the proof of Lemma~\ref{lemma::medium}, using the spectral norm bound 
on $X^\perp$ established above, gives $\frac{2}{nT}|\mathrm{Tr}[(X^\parallel - 
G)'X^\perp]| = o_P(1)$. Therefore
\[
\min_{\mathrm{rank}(G) \leq 2R} \frac{1}{nT}\|X - G\|_2^2 \geq c + o_P(1).
\]

\medskip\noindent\textbf{Low-rank approximations.}
These are directly assumed in condition (iii).

\medskip\noindent\textbf{Conclusion.}
All conditions of Theorems~\ref{thm::PC_consistencyNEW} 
and~\ref{thm::IFE_consistencyNEW} have been verified conditionally on $(U, V)$. 
Therefore,
\[
\widehat{\beta}_{\rm PC} = \beta^*_{nT} + o_P(1), \qquad 
\widehat{\beta}_{\rm IFE} = \beta^*_{nT} + o_P(1).
\]
\end{proof}

\end{document}